\documentclass[12pt]{article}
\usepackage{amsfonts,amssymb,amsthm}

\voffset = - 1.2cm
\def\Pro{{\rm \Pi}}
\def\degr{{\rm deg}}
\def\ham{{\cal {H}}_\Lambda}
\def\hamb{{\cal {H}}}
\def\grupo{{\cal {G}}_\Lambda}
\def\map{{\cal {T}}}
\def\evo{{\cal {E}}_0}
\def\rr{{\rho}}
\def\mm{{r}}
\def\JJ{{I}}
\def\KK{{J}}
\def\LL{{K}}
\def\casU{{\cal C}} 
\def\casS{{C}} 
\def\ve{{v}}
\def\FF{{F}}
\def\aA{{A}}
\def\bB{{B}}
\def\cC{{\Upsilon}}
\def\MM{{\mu}}

\theoremstyle{plain}
\newtheorem{theorem}{Theorem}
\newtheorem{corollary}[theorem]{Corollary}
\newtheorem{lemma}[theorem]{Lemma}
\newtheorem{proposition}[theorem]{Proposition}

\theoremstyle{definition}
\newtheorem{definition}[theorem]{Definition}

\begin{document}

\centerline{\Large \bf From constants of motion to superposition rules}
\vskip 0.25cm
\centerline{\Large \bf for Lie--Hamilton systems}
\vskip 0.5cm
\centerline{A Ballesteros$^1$, J F Cari\~nena$^2$, F J Herranz$^1$, J de Lucas$^3$ and C~Sard\'on$^4$}
\vskip 0.5cm

\centerline{$^1$Department of Physics, University of Burgos, 09.001 Burgos, Spain.}
\vskip 0.25cm
\centerline{$^2$Department of Theoretical Physics and IUMA,
University of Zaragoza,}
\centerline {c. Pedro Cerbuna 12, 50.009 Zaragoza, Spain.}
\vskip 0.25cm
\centerline{$^3$Faculty of Mathematics and Natural Sciences,
Cardinal Stefan}
\centerline{ Wyszy\'nski University, ul. W\'oycickiego 1/3, 01-938 Warsaw, Poland.}
\vskip 0.25cm
\centerline{$^4$Department of Fundamental Physics, Faculty of
Sciences,}
\centerline{ University of Salamanca, Plaza de la Merced s/n, 37.008 Salamanca, Spain.}
\vskip 0.25cm

\begin{abstract}
A {\it Lie system} is a nonautonomous system of first-order differential equations possessing a {\it superposition rule}, i.e. a map
expressing its general solution in terms of a generic finite family of particular solutions and some constants.
Lie--Hamilton systems form a subclass of Lie systems whose dynamics is governed by
a curve in a finite-dimensional real Lie algebra of functions on a Poisson
manifold. It is shown that Lie--Hamilton systems are naturally endowed with a Poisson coalgebra structure. This allows us to devise methods to derive in an algebraic way their constants of motion and superposition rules. We illustrate our methods by studying Kummer--Schwarz equations,  Riccati equations, Ermakov systems and Smorodinsky--Winternitz systems with time-dependent frequency.
\end{abstract}


{\bf AMS:} {34A26 (primary) 70G45, 70H99 (secondary)}\\



\section{Introduction}
The importance of geometric methods for the study of differential equations is
unquestionable. They have led to the description of many
interesting properties of differential equations, which in turn have been
applied in remarkable mathematical and physical problems
\cite{Ol93}-\cite{Be84}. Among geometric methods, we here focus on those of the theory of Lie systems
\cite{LS}-\cite{Dissertationes}.

Lie systems have lately been analysed thoroughly (see \cite{Dissertationes} and
references therein). This
has given rise to new techniques that have been employed to
study
relevant differential equations occurring in physics \cite{SIGMA}-\cite{FLV10},
mathematics \cite{CGL12}, control theory \cite{Clem06}, economy \cite{cal}, etc.

The celebrated {\it Lie--Scheffers Theorem} \cite{LS,CGM07}
states that a Lie system amounts to a $t$-dependent vector field taking values in a
finite-dimensional real Lie algebra of vector fields, a so-called {\it
Vessiot--Guldberg Lie algebra} of the system \cite{LS,CGM07}. In this work we concentrate on
analysing a subclass of Lie systems on Poisson manifolds, the
referred to as {\it Lie--Hamilton systems} \cite{CLS12Ham}, that admit a Vessiot--Guldberg Lie algebra of Hamiltonian vector fields
with
respect to a Poisson structure.

As an example of Lie--Hamilton systems, consider the  first-order system of differential equations on $\mathcal{O}=\{(x,p)\in{\rm T}^*\mathbb{R}\mid p< 0\}$
of the form
\begin{equation}\label{SecOrdRic}
\left\{
\begin{array}{ll}
\frac{{\rm d} x}{{\rm d} t}&=\frac{1}{\sqrt{-p}}-a_0(t)-a_1(t)x-a_2(t)x^2,\\
\frac{{\rm d} p}{{\rm d} t}&= p(a_1(t)+2a_2(t)x),
\end{array}\right.
\end{equation}
with $a_0(t),a_1(t),a_2(t)$ being arbitrary functions, occurring in the
study of some second-order Riccati equations \cite{CLS12Ham,CLS12}. This system is a
Lie--Hamilton system as it describes the integral curves of the $t$-dependent
vector field
\[
X_t=X_1-a_0(t)X_2-a_1(t)X_3-a_2(t)X_4
,\]
satisfying that the vector fields $\{X_t\}_{t\in\mathbb{R}}$ are contained in the five-dimensional real Lie algebra of vector fields on  $\mathcal{O}$  spanned by 
\begin{eqnarray}\label{VF}
\qquad X_1=\frac{1}{\sqrt{-p}}\frac{\partial}{\partial x},\quad
X_2=\frac{\partial}{\partial x},\quad
X_3=x\frac{\partial}{\partial x}-p\frac{\partial}{\partial p},\quad
X_4=x^2\frac{\partial}{\partial x}-2xp\frac{\partial}{\partial p},\nonumber \\   \qquad
X_5=\frac{x}{\sqrt{-p}}\frac{\partial}{\partial
x}+2\sqrt{-p}\frac{\partial}{\partial p},
\end{eqnarray}
which are additionally Hamiltonian vector fields
relative to the Poisson bivector
$\Lambda=\partial/\partial x\wedge\partial/\partial p$ on $\mathcal{O}$. Indeed, they admit the 
Hamiltonian functions 
\begin{eqnarray}\label{equFun}
h_1=-2\sqrt{-p}, \quad  h_2=p,\quad h_3=xp,\quad h_4=x^2p,\quad  h_5=-2x\sqrt{-p},
\end{eqnarray}
which span along with $h_6=1$ a six-dimensional real Lie algebra of functions with respect to the Poisson bracket induced
by $\Lambda$ (see \cite{CLS12} for details).

  Lie--Hamilton systems can be employed to investigate remarkable dynamical sytems and enjoy a plethora of geometric properties \cite{CLS12Ham,ADR11,BBHMR09,Ru10}. For example,
the Vessiot--Guldberg Lie algebras of Hamiltonian vector fields associated to these systems
give rise to $t$-dependent Hamiltonians that can be understood as curves in finite-dimensional
real Lie algebras of functions, the hereafter called {\it Lie--Hamilton algebras} \cite{CLS12Ham}.
This property has recently been employed to study $t$-independent constants of motion of
Lie--Hamilton systems \cite{CLS12Ham}. Note that $t$ generally stands for the time when referring to Lie--Hamilton
systems describing a physical system.

In this work, we first employ Lie--Hamilton algebras so as to analyse the
constants of motion for Lie--Hamilton systems. This provides some generalisations of the results given in
\cite{CLS12Ham} about $t$-independent constants of motion and in \cite{Ru10} about $t$-dependent constants of motion for Lie--Hamilton systems. In addition, we demonstrate that several types of constants of motion of Lie--Hamilton systems form a Poisson algebra, and we devise algebraic procedures to derive them or, at least, to
simplify their calculation.

 We propose the use of {\it Poisson coalgebra} techniques to obtain superposition rules for Lie--Hamilton systems in an algebraic manner. We recall that Poisson coalgebras are Poisson algebras
endowed with a Poisson algebra homomorphism called the coproduct map~\cite{CP}. We show that the coproduct map applied onto a Casimir function of an appropriate Poisson coalgebra naturally related to a Lie--Hamilton system $X$ gives rise to $t$-independent constants of motion that can be employed to study the superposition rules and the $N$-dimensional generalisations of $X$.

Since our procedures are algebraic, they are much simpler than standard methods \cite{PW,CGM07}, which require the integration of  systems of partial/ordinary
differential equations.
As an application, we derive constants of motion and
superposition rules for some Lie--Hamilton systems of interest: Kummer--Schwarz equations
in Hamiltonian form \cite{CLS12Ham},  systems of Riccati equations \cite{CGM07}, Smorodinsky--Winternitz systems with time-dependent frequency~\cite{WSUF65,WSUF67}, 
and some classical mechanical systems \cite{ADR11,Pi50,Er08}.

The structure of the paper goes as follows. We describe  the conventions and the most fundamental notions to be used throughout our
paper in Section 2. In Section 3, we introduce the fundamental
features of Lie--Hamilton systems. In Section 4, we analyse the algebraic structure of
general constants of motion for Lie--Hamilton systems. Subsequently, 
we focus on the study of several relevant particular types of such constants  in Section 5. This
involves the description of some new methods for their calculation.  In Section 6, Poisson coalgebras are shown to provide new simpler methods to derive constants of motion and superposition rules
for Lie--Hamilton systems. Several examples are analysed in Section 7. Finally, we summarise our results and plan of future research in Section 8.

\section{Preliminaries}\label{LSLS}

In this section we survey $t$-dependent vector fields \cite{Dissertationes}, Poisson algebras \cite{CLS12Ham} and Poisson coalgebra structures~\cite{CP,BR}. For simplicity, we generally assume functions and geometric structures
to be real, smooth, and globally defined. This permits us to omit minor
technical problems so as to highlight the main aspects of our results.

A Lie algebra is a  pair $(V,[\cdot,\cdot])$, where $V$ stands for a
real linear space equipped with a Lie
 bracket $[\cdot\,,\cdot]:V\times V\rightarrow V$. We define ${\rm
Lie}(\mathcal{B},V,[\cdot,\cdot])$ to be the smallest Lie subalgebra of $(V,[\cdot,\cdot])$ containing $\mathcal{B}$. When its meaning is clear, we write $V$ and  ${\rm Lie}(\mathcal{B})$ instead of $(V,[\cdot,\cdot])$ and ${\rm
Lie}(\mathcal{B},V,[\cdot,\cdot])$, respectively, 

A {\it $t$-dependent vector field} on    a manifold $N$ is a map $X:(t,x)\in\mathbb{R}\times
N\mapsto X(t,x)\in {\rm T}N$
such that $\tau_N\circ X=\pi_2$, with $\pi_2:(t,x)\in\mathbb{R}\times N\mapsto
x\in N$ and $\tau_N:{\rm T}N\rightarrow N$ being the projection associated to the
tangent bundle ${\rm T}N$. This amounts to saying that $X$
is a $t$-parametrized family of standard vector fields $\{X_t\}_{t\in\mathbb{R}}$, with $X_t:x\in N\mapsto
X(t,x)\in {\rm T}N$. We call {\it minimal Lie algebra} of $X$ the smallest real Lie algebra of vector fields, $V^X$, containing
$\{X_t\}_{t\in\mathbb{R}}$, namely $V^X={\rm
Lie}(\{X_t\}_{t\in\mathbb{R}})$. We can also consider each $X$ as
a vector field on $\mathbb{R}\times N$
 by defining $(Xf)(t,x)=(X_tf_t)(x)$, where $f\in C^\infty(\mathbb{R}\times N)$
and $f_t\in C^\infty(N)$ stands for $f_t:x\in N\mapsto f(t,x)\in\mathbb{R}$.

We call {\it integral curves} of $X$ the integral
curves of its {\it
suspension}, i.e.~the vector field $\bar
X(t,x)=\partial/\partial
t+X(t,x)$ on $\mathbb{R}\times N$ \cite{Dissertationes,Foundations,CR89}. Every integral curve
$\gamma:t\in\mathbb{R}\mapsto (t,x(t))\in\mathbb{R}\times N$ of $X$ satisfies the {\it associated system} to $X$, i.e.
\begin{equation}\label{aso}
\frac{{\rm d}(\pi_2 \circ \gamma)}{{\rm d} t}(t)=(X\circ \gamma)( t).
\end{equation}
Conversely, there exists a unique
$X$ whose integral curves of the form $(t,x(t))$ describe
the solutions of a system of first-order systems in normal form.
This establishes a bijection between $t$-dependent vector fields and such systems, which justifies to use $X$ to designate both $X$ and (\ref{aso}).

 A {\it Poisson algebra}  is a triple $(A,\star, \{\cdot,\cdot\})$
consisting of an $\mathbb{R}$-linear space $A$ along with a
 product $\star:A\times A\rightarrow A$ so that $(A,\star)$ becomes an
associative $\mathbb{R}$-algebra, and a Lie bracket $\{\cdot,\cdot\}$ on $A$,
the so-called {\it Poisson bracket}, such that $(A,\{\cdot,\cdot\})$ is
a (possibly infinite-dimensional) Lie algebra and
\[
\{b\star c,a\}=b\star \{c,a\}+\{b,a\}\star c,\qquad \forall a,b,c\in A.
\]
The above compatibility condition, the called {\it Leibnitz rule},
amounts to saying that the Poisson bracket is a derivation of $(A,\star)$ on each factor. Contrarily to the usual
convention, the product $\star$ of our definition of Poisson algebras may be non-commutative, which simplifies posterior definitions.
We call {\it Casimir element} of $A$ an $a\in A$ such that $\{c,a\}=0$ for all $c\in A$. The set ${\rm Cas}(A)$ of Casimir elements of $A$
is an ideal of $(A,\{\cdot,\cdot\})$. For simplicity, we will sometimes write $ab$ for the  $a\star b$. A {\it Poisson algebra morphism} is a morphism $\map:(A,\star_A,\{\cdot,\cdot\}_A)\rightarrow (B,\star_B,\{\cdot,\cdot\}_B)$ of $\mathbb{R}$-algebras
$\map:(A,\star_A)\rightarrow (B,\star_B)$ that also satisfies that $\map(\{a,b\}_A)=\{\map(a),\map(b)\}_B$ for every $a,b\in A$. 

One of the types of Poisson algebras to be used in this paper are symmetric and universal algebras. Let us describe their
main characteristics. Given a finite-dimensional real Lie algebra $(\mathfrak{g},[\cdot,\cdot]_\mathfrak{g})$, its universal algebra, $U_\mathfrak{g}$, is obtained from the quotient
$T_\mathfrak{g}/\mathcal{R}$ of the tensor algebra $(T_\mathfrak{g},\otimes)$ of $\mathfrak{g}$
by the bilateral ideal  $\mathcal{R}$ spanned by the elements
$v\otimes w-w\otimes v-[v,w]$, with $v,w\in \mathfrak{g}$. Given the quotient map $\pi:T_\mathfrak{g}\rightarrow U_\mathfrak{g}$, the space $U_\mathfrak{g}$ becomes an $\mathbb{R}$-algebra  $(U_\mathfrak{g},\widetilde \otimes)$ when endowed with the product $\widetilde \otimes: U_\mathfrak{g}\times U_\mathfrak{g}\rightarrow U_\mathfrak{g}$  given by
$\pi(P)\,\widetilde\otimes\,\pi(Q)\equiv \pi(P\otimes Q)$, for every $P,Q\in T_\mathfrak{g}$. The Lie
bracket on $\mathfrak{g}$ can be extended to a Lie bracket $\{\cdot,\cdot\}_{U_\mathfrak{g}}$ on
$U_\mathfrak{g}$ by imposing it to be a derivation of
$(U_\mathfrak{g},\widetilde\otimes)$ on each factor.  This turns
$U_\mathfrak{g}$ into a Poisson algebra $(U_\mathfrak{g},\widetilde\otimes,\{\cdot,\cdot\}_{U_\mathfrak{g}})$ \cite{CL99}. The elements
of its Casimir subalgebra are henceforth dubbed as {\it Casimir elements} of $\mathfrak{g}$ \cite{Be81}.

If we set $\mathcal{R}$ to be the bilateral ideal spanned by the elements $v\otimes w-w\otimes v$  in the above procedure, we obtain a new commutative Poisson algebra $S_\mathfrak{g}$ called
{\it symmetric algebra} of $(\mathfrak{g},[\cdot,\cdot]_\mathfrak{g})$. The elements of $S_\mathfrak{g}$ are polynomials on the elements of $\mathfrak{g}$. Via the isomorphism 
$\mathfrak{g}\simeq (\mathfrak{g}^*)^*$,
they can naturally be understood as polynomial functions on $\mathfrak{g}^*$ \cite{CL99,AMA75}. The Casimir elements of
this Poisson algebra are called {\it Casimir functions} of $\mathfrak{g}$.

The Poisson algebras $U_\mathfrak{g}$ and $S_\mathfrak{g}$ are related by the
{\it symmetrizer map} \cite{Be81,AMA75,Var98}, i.e. the linear
isomorphism $\lambda:S_\mathfrak{g}\rightarrow 
U_\mathfrak{g}$ of the form
\begin{equation}
\label{symmap}
\lambda(v_{i_1})=\pi(v_{i_1}),\quad
\lambda(v_{i_1}v_{i_2}\ldots v_{i_l})=\frac{1}{l!}\sum_{s\in
\Pi_l}\lambda(v_{s(i_1)})\, \widetilde\otimes\ldots\widetilde\otimes\, \lambda (v_{s(i_l)}),
\end{equation}
for all $v_{i_1},\ldots,v_{i_l}\in\mathfrak{g}$ and with $\Pi_l$ being the set of permutations of $l$ elements. Moreover, 
\begin{equation}\label{UgSg}\nonumber
\lambda^{-1} \left(\{ v,P\}_{U_\mathfrak{g}} \right)=\{v,\lambda^{-1}(P)\}_{S_\mathfrak{g}}, \qquad \forall P\in
U_\mathfrak{g},\quad\forall v\in\mathfrak{g}.
\end{equation}
So, $\lambda^{-1}$ maps the Casimir elements of $\mathfrak{g}$ into Casimir elements of $S_\mathfrak{g}$.

If $(A,\star_A,\{\cdot,\cdot\}_A)$ and $(B,\star_B,\{\cdot,\cdot\}_B)$ are Poisson algebras and $\star_A,$ $\star_B$ are commutative, then
$A\otimes B$ becomes a Poisson algebra $(A\otimes B,\star_{A\otimes B},\{\cdot,\cdot\}_{A\otimes
B})$ by defining
\begin{eqnarray*}
  (a\otimes b)\star_{A\otimes B} (c\otimes d)&=(a\star_A c)\otimes (b\star_B d),\qquad \forall a,c\in A,\quad \forall b,d\in B,\\
  \{a\otimes b, c\otimes d\}_{A\otimes B}&=\{a, c\}_A\otimes b\star_B d+a\star_A
c\otimes\{b,d\}_B .
\end{eqnarray*}
Similarly, a Poisson structure on $A^{(m)}\equiv\   \stackrel{m-{\rm times}}{\overbrace{  A\otimes \ldots\otimes A}}$ can be constructed by induction.

We say that $(A,\star_A,\{\cdot,\cdot\}_A,\Delta)$ is a {\it Poisson coalgebra} if $(A,\star_A,\{\cdot,\cdot\}_A)$ is a Poisson
algebra and $\Delta:(A,\star_A,\{\cdot,\cdot\}_A)\rightarrow (A\otimes
A,\star_{A\otimes A},\{\cdot,\cdot\}_{A\otimes A})$, the so-called {\it coproduct}, is a
Poisson algebra
homomorphism
which is  {\it coassociative} \cite{CP}, i.e. $
(\Delta \otimes {\rm Id}) \circ \Delta=({\rm Id} \otimes \Delta) \circ \Delta$.
Then, the {$m$-th coproduct} map
$\Delta ^{(m)}:  A\rightarrow  A^{(m)}$
can be defined recursively as follows
\begin{equation}\label{copr}
{\Delta}^{(m)}= ({\stackrel{(m-2)-{\rm times}}{\overbrace{{\rm
Id}\otimes\ldots\otimes{\rm Id}}}}\otimes {\Delta^{(2)}})\circ \Delta^{(m-1)},\qquad m>2,
\end{equation}
and $\Delta\equiv \Delta^{(2)}$. Such an induction ensures that $\Delta^{(m)}$ is also a Poisson map.

In particular, $S_\mathfrak{g}$ is a Poisson coalgebra with {\em primitive coproduct map} given by
$\Delta(v)= v\otimes 1+1\otimes v,$ for all $v\in\mathfrak{g}\subset S_\mathfrak{g}$.
The coassociativity of $\Delta$ is straightforward, and its $m$-th generalization reads
\[
  \Delta^{(m)}(v)=v\otimes{ \stackrel{(m-1)-{\rm times}}{\overbrace{{\rm
1}\otimes\ldots\otimes{\rm 1}}}}
 +\, {\rm 1}\otimes v   \otimes{ \stackrel{(m-2)-{\rm times}}{\overbrace{{\rm
1}\otimes\ldots\otimes{\rm 1}}}}  +
\ldots
  +  { \stackrel{(m-1)-{\rm times}}{\overbrace{{\rm
1}\otimes\ldots\otimes{\rm 1}}}}  \otimes v, \quad \forall v\in\mathfrak{g}\subset S_\mathfrak{g}.
\]

A {\it Poisson manifold} is a pair $(N,\{\cdot,\cdot\})$ such that
$(C^\infty(N),\star, \{\cdot,\cdot\})$, where ``$\star$''
 stands for the standard product of
functions, is a Poisson algebra. We call $\{\cdot,\cdot \}$ the {\it Poisson structure} of the Poisson manifold.
For instance,
the dual space $\mathfrak{g}^*$ of a Lie algebra
$(\mathfrak{g},[\,\cdot,\cdot\,]_{\mathfrak{g}})$ can be  endowed
with the {\it Lie--Poisson structure} \cite{Va94}, namely
\[
\{f,g\}(\theta)=\langle [df_\theta,d g_\theta]_{\mathfrak{g}},\theta\rangle,\qquad f,g\in
C^\infty(\mathfrak{g}^*),\qquad \theta\in\mathfrak{g}^*.
\]

As every $\{\cdot,f\}$, with $f\in C^\infty(N)$, is a derivation on
$(C^\infty(N),\star)$, there exists a unique vector field $X_f$ on $N$, the
referred to as {\it Hamiltonian vector field} associated with $f$,
such that $X_fg=\{g,f\}$ for all $g\in C^{\infty}(N)$. The Jacobi identity  for
the Poisson structure then entails $X_{\{f,g\}}=-[X_f,X_g]$,  $\forall f,g\in
C^{\infty}(N)$.
Hence, the mapping $f\mapsto X_f$ is a Lie algebra anti-homomorphism
from $(C^{\infty}(N),\{\cdot,\cdot\})$ to $(\Gamma({\rm T}N),[\cdot,\cdot])$, where $\Gamma({\rm T}N)$ is  the space of sections of the tangent bundle to $N$.

As every Poisson structure is a derivation in each entry, it determines a unique
bivector field $\Lambda\in \Gamma(\bigwedge^2 {\rm T}N)$ such that
\begin{equation}\label{ForPoi}
\{f,g\}=\Lambda(df,dg),\qquad \forall f,g\in C^{\infty}(N).
\end{equation}
We call $\Lambda$ the {\it Poisson bivector} of the Poisson manifold
$(N,\{\cdot,\cdot\})$.
In view of the Jacobi identity for $\{\cdot,\cdot\}$, it follows that
$[\Lambda,\Lambda]_{\rm SN}=0$,
with $[\cdot,\cdot]_{\rm SN}$ being the {\it Schouten--Nijenhuis Lie bracket} \cite{Va94}.
Conversely, every bivector field
 $\Lambda$ on $N$ satisfying this condition gives rise to a Poisson
structure by formula (\ref{ForPoi}).
This justifies to refer to Poisson manifolds as $(N,\{\cdot,\cdot\})$ or
$(N,\Lambda)$ indistinctly.

Observe that $(N,\Lambda)$ induces a unique bundle morphism
$\widehat\Lambda:{\rm T}^*N\rightarrow {\rm T}N$ satisfying
 $\theta'(\widehat \Lambda(\theta))=\Lambda(\theta,\theta')$ for every $\theta, \theta'\in \Gamma(T^*N)$.
 So, $X_f=-\widehat \Lambda (df)$, for any  $f\in C^\infty(N)$.

\section{Lie systems and Lie--Hamilton systems}\label{LSLH}

Recall that Lie systems are systems of first-order differential equations
admitting a superposition rule (see \cite{CGM07,Dissertationes} for details).

\begin{definition} A {\it superposition rule} for a system $X$ on $N$ is a function
$\Phi:N^{m}\times N\rightarrow
N$ of the form $x=\Phi(x_{(1)}, \ldots,x_{(m)};k)$ allowing us to write
the general
 solution $x(t)$ of $X$ as
  $x(t)=\Phi(x_{(1)}(t), \ldots,x_{(m)}(t);k),$
where $x_{(1)}(t),\ldots,x_{(m)}(t)$ is a generic family of
particular solutions and $k$ is a point of $N$.
 \end{definition}

The conditions ensuring that a system $X$ possesses a superposition rule are
described by the {\it Lie--Scheffers Theorem} \cite{LS,CGM07,Dissertationes}.

\begin{theorem} A system $X$ on $N$ admits a superposition rule if and only if $X_t=\sum_{\alpha=1}^rb_\alpha(t)X_\alpha$
for a family $b_1(t),\ldots,b_r(t)$  of $t$-dependent functions and a
set of vector fields $X_1,\ldots,X_r$ on $N$ spanning an $r$-dimensional real Lie algebra,
a so-called  Vessiot--Guldberg Lie algebra associated to
$X$~\cite{CRC96}.
In other words, $X$ is a Lie system if and only if $V^X$ is finite-dimensional.
\end{theorem}

In order to illustrate the above concepts, let us consider  a {\it Riccati equation}, i.e. 
\begin{equation}\label{Ric}
\frac{{\rm d} x}{{\rm d} t}=a_0(t)+a_1(t)x+a_2(t)x^2,\\
\end{equation}
where $a_0(t),a_1(t),a_2(t)$ are arbitrary functions. Observe that (\ref{Ric}) is  the system associated to
the $t$-dependent vector field $X=a_0(t)X_1+a_1(t)X_2+a_2(t)X_3$, where
\[
X_1=\frac{\partial}{\partial x},\qquad X_2=x\frac{\partial}{\partial x},\qquad X_3=x^2\frac{\partial}{\partial x}
\]
span a Vessiot--Guldberg Lie algebra $V$ for (\ref{Ric}) isomorphic to $\mathfrak{sl}(2,\mathbb{R})$ \cite{PW}. The Lie--Scheffers
Theorem shows that Riccati must admit a superposition rule. Indeed,
the general solution of (\ref{Ric}) can be brought into the form
$$
x(t)=\frac{x_1(t)(x_3(t)-x_2(t))+kx_3(t)(x_1(t)-x_2(t))}{(x_3(t)-x_2(t))+k(x_1(t)-x_2(t))},
$$
where $x_{(1)}(t),x_{(2)}(t),x_{(3)}(t)$ are three different particular
solutions and $k$ is an arbitrary constant.
So, Riccati equations possess a superposition rule
$\Phi:(x_{(1)},x_{(2)},x_{(3)};k)\in\mathbb{R}^3\times\mathbb{R}\mapsto
x\in\mathbb{R}$ given by
\[
x=\frac{x_1(x_3-x_2)+kx_3(x_1-x_2)}{(x_3-x_2)+k(x_1-x_2)}\]
enabling us to write their general solutions as $x(t)=\Phi(x_{(1)}(t),x_{(2)}(t),x_{(3)}(t);k)$.

There exist several ways to derive superposition rules for Lie systems
\cite{CGM07,AHW81}. The method employed in this work (see \cite[p.~4]{CGL12} and \cite{CL11Sec} for a detailed description and examples) is based on the so-called {\it diagonal
prolongations}.

\begin{definition} Given a $t$-dependent vector field $X$ on $N$, its  {\it
diagonal prolongation} to $N^{m+1}$ is the unique $t$-dependent vector field
$\widetilde X$ on $N^{m+1}$ such that:
\begin{enumerate}
\item Each vector field $\widetilde X_t$ projects onto $X_t$ via $\Pro:(x_{(0)},\ldots,x_{(m)})\in
N^{m+1}\mapsto x_{(0)}\in N$.
\item $\widetilde X$ is invariant under permutations of variables
$x_{(i)}\leftrightarrow x_{(j)}$, with $i,j=0,\ldots,m$.
\end{enumerate}
\end{definition}

Roughly speaking, we can derive superposition rules by obtaining a certain set of constants of motion for an appropriate
diagonal prolongation
of a Lie system. In the literature, this is usually performed by solving a system of PDEs by the method of characteristics
(cf.~\cite{CLS12Ham}), which  can be quite laborious~\cite{CL11Sec}. As seen in this work,
this task can be simplified in the case of {\it Lie--Hamilton systems} through our methods.

\begin{definition} A system $X$ on $N$ is said to be a {\it Lie--Hamilton system} if
$N$ can be endowed with a Poisson bivector $\Lambda$ in such a way that $V^X$ becomes a finite-dimensional real
Lie algebra of Hamiltonian vector fields relative to $\Lambda$.
\end{definition}

Lie--Hamilton systems enjoy relevant features, e.g. (\ref{SecOrdRic}) admits a $t$-dependent Hamiltonian
\begin{equation}\label{Ex}
 h(t)=h_1-a_0(t)h_2-a_1(t)h_3-a_2(t)h_4.
 \end{equation}
The functions $h_1,\ldots,h_5$ and $h_6=1$ given by (\ref{equFun})  span a six-dimensional real Lie algebra $(\mathfrak{h}_6,\{\cdot,\cdot\}_\Lambda)$.
It is also relevant that $\mathfrak{h}_6\simeq \mathfrak{sl}(2,\mathbb{R})
\oplus_s \mathfrak{h}_3$,
where $\mathfrak{sl}(2,\mathbb{R}) \simeq \langle h_2,h_3,h_4\rangle $,
$\mathfrak{h}_3\simeq  \langle h_1,h_5,h_6\rangle$ is
the radical of $\mathfrak{h}_6$, which is isomorphic to the Heisenberg--Weyl Lie algebra, and
$\oplus_s$ stands for a semidirect sum. Additionally, $\mathfrak{h}_6$ is also isomorphic to the  two-photon  Lie algebra and the  $(1+1)$-dimensional centrally extended Schr\"odinger  Lie algebra    (cf.~\cite{Schrod,BBF09}).

The $t$-dependent Hamiltonians of the form (\ref{Ex}) are called {\it
Lie--Hamiltonian structures}. More specifically, we have the following definition.

\begin{definition} A {\it Lie--Hamiltonian structure} is a triple
$(N,\Lambda,h)$, where $(N,\Lambda)$ stands for  a Poisson manifold and $h:
(t,x)\in \mathbb{R}\times N\mapsto h_t(x)=h(t,x)\in  N$ is such that
 the space $(\ham\equiv {\rm Lie}(\{h_t\}_{t\in\mathbb{R}},\{\cdot,\cdot\}_\Lambda),\{\cdot,\cdot\}_\Lambda)$ is 
a finite-dimensional
real Lie algebra.
\end{definition}

We now describe some features of Lie--Hamilton systems (see \cite{CLS12Ham} for details).

\begin{theorem} A system $X$ on $N$ is a {\it Lie--Hamilton system} if and only
if there exists a Lie--Hamiltonian structure $(N,\Lambda,h)$ such that $X_t=
-\widehat \Lambda(dh_t)$ for every $t\in\mathbb{R}$. In this case, we call $(\ham,\{\cdot,\cdot\}_\Lambda)$ a {\it Lie--Hamilton} algebra for
$X$.
\end{theorem}

\begin{proposition}\label{IntLieHam0} Given a Lie--Hamilton system $X$ on $N$ admitting a Lie--Hamiltonian structure $(N,\Lambda,h)$, a $t$-independent function is a constant of
motion for $X$ if and only if it Poisson commutes with the elements of $\ham$. The  family $\mathcal{I}^X$ of $t$-independent constants of
motion of $X$ form a Poisson algebra $(\mathcal{I}^X,\cdot,\{\cdot,\cdot\}_\Lambda)$.
\end{proposition}

\section{On constants of motion for Lie--Hamilton systems}
The aim of this section is to present an analysis of the algebraic properties
 of the  constants of motion for Lie--Hamilton systems.
More specifically,
we prove that a Lie--Hamiltonian structure $(N,\Lambda,h)$ for a
Lie--Hamilton system $X$ induces a Poisson
bivector on $\mathbb{R}\times N$. This allows us
to endow the space of constants of motion for $X$ with a Poisson algebra structure, which can be used to produce new constants of motion from known ones.  Our achievements extend
to general constants of motion the results derived for $t$-independent constants of motion for Lie--Hamilton systems  in \cite{CLS12Ham}.

Given a system $X$ on $N$, a constant of motion for $X$ is a first-integral $f\in C^\infty(\mathbb{R}\times N)$ of the autonomisation  $\bar X$ of $X$, namely
\begin{equation}\label{EqIn}
\frac{\partial f}{\partial t}+Xf=\bar Xf=0,
\end{equation}
where $X$ is understood as a vector field on $\mathbb{R}\times N$. Using this, we can straightforwardly prove the following proposition.

\begin{proposition}\label{Prop:Ralg} The space $\bar{\mathcal{I}}^X$ of $t$-dependent constants of motion for a
system $X$ forms an $\mathbb{R}$-algebra $(\bar{\mathcal{I}}^X,\cdot)$.
\end{proposition}

To generalise the second statement of Proposition \ref{IntLieHam0}   to $t$-dependent
 constants of motion, we endow $\mathbb{R}\times N$ with
a Poisson structure that makes $\bar{\mathcal{I}}^X$ into a Poisson algebra.
 
\begin{lemma} Every Poisson manifold $(N,\Lambda)$ induces a Poisson manifold $(\mathbb{R}\times N,\bar\Lambda)$ with Poisson structure
\begin{equation}\label{newBrack}\nonumber
\{f,g\}_{\bar\Lambda}(t,x)\equiv\{f_t,g_t\}_{\Lambda}(x),\qquad (t,x)\in
\mathbb{R}\times N.
\end{equation}
\end{lemma}

\begin{definition} Given a Poisson manifold $(N,\Lambda)$,
the associated Poisson manifold $(\mathbb{R}\times N,\bar \Lambda)$ is called the {\it autonomisation} of $(N,\Lambda)$. Likewise,
 the Poisson bivector $\bar {\Lambda}$ is called the {\it autonomisation} of $\Lambda$.
\end{definition}

The following lemma allows us to prove that $(\bar{\mathcal{I}}^X,\cdot,\{\cdot,\cdot\}_{\bar\Lambda})$
is a Poisson algebra.

\begin{lemma}\label{AutLieHam} Let $(N,\Lambda)$ be a Poisson manifold and
$X$ be a Hamiltonian vector field on $N$ relative to $\Lambda$. Then,
$\mathcal {L}_{\bar{X}}\bar{\Lambda}=0$.
\end{lemma}
\begin{proof}
Given a coordinate system $\{x_1,\ldots,x_n\}$ for $N$ and $x_0 \equiv t$ in $\mathbb{R}$, we can naturally define a coordinate system $\{x_0,x_1,\ldots,x_n\}$ on $\mathbb{R}\times N$.
Since $(x_0)_t=t$ is constant as a function on $N$, then $\bar\Lambda(dx_0,df)=\{(x_0)_{t},f_t\}_\Lambda=0$ for every $f\in C^\infty(\mathbb{R}\times N)$. Additionally, 
$(x_i)_t=x_i$ for $i=1,\ldots,n$. Hence, we have that $\bar\Lambda({\rm d} x_i,{\rm d} x_j)=\{x_i,x_j\}_{\bar\Lambda}$ is a $x_0$-independent function for $i,j=0,\ldots,n$. So,	
\[
(\mathcal{L}_{\bar{X}}\bar{\Lambda})(t,x)=\left[\mathcal{L}_{\frac{\partial}{\partial x_0}+X}\left(\sum_{i<j=1}^n\{x_i,x_j\}_\Lambda\frac{\partial}{\partial x_i}\wedge\frac{\partial}{\partial x_j}\right)\right](x)=(\mathcal{L}_{X}\Lambda)(x).
\]
Since $X$ is Hamiltonian, we obtain $0=(\mathcal{L}_{X}\Lambda)(x)=(\mathcal{L}_{\bar{X}}\bar\Lambda)(t,x)=0$.
\end{proof}

Now, we can establish the main result of this section.

\begin{proposition}\label{AutLieHam2}  Let $X$ be a Lie--Hamilton system on $N$
possessing a Lie--Hamiltonian structure $(N,
\Lambda,h)$, then the space $(\overline{\mathcal{I}}^X,\cdot,\{\cdot,\cdot\}_{\bar\Lambda})$
is a Poisson algebra.
\end{proposition}
\begin{proof} From Proposition \ref{Prop:Ralg}, we see that
$(\bar{\mathcal{I}}^X,\cdot)$ is an $\mathbb{R}$-algebra. To demonstrate that $(\bar{\mathcal{I}}^X,\cdot,\{\cdot,\cdot\}_{\bar\Lambda})$ is a Poisson algebra, it remains to prove that the Poisson bracket of any two elements $f,g$ of $\bar{\mathcal{I}}^X$ remains in it, i.e.~$X\{f,g\}_{\bar{\Lambda}}=0$.
By taking into account that the vector fields $\{X_t\}_{t\in\mathbb{R}}$ are Hamiltonian relative to $(N,\Lambda)$ and
Lemma \ref{AutLieHam}, we find that $\bar\Lambda$ is invariant under the autonomization of each vector field $X_{t'}$ with $t'\in\mathbb{R}$, i.e. $
\mathcal{L}_{\overline{X_{t'}}}\bar\Lambda=0$. Therefore,
\begin{eqnarray}
  \bar X\{f,g\}_{\bar{\Lambda}}(t',x)&=\overline {X_{t'}}\{f,g\}_{\bar{\Lambda}}(t',x)= \{\overline{X_{t'}}f,g\}_{\bar{\Lambda}}(t',x)+
\{f,\overline{ X_{t'}}g\}_{\bar{\Lambda}}(t',x)=\\ &=\{\bar{X}f,g\}_{\bar{\Lambda}}(t',x)+
\{f,\bar{X}g\}_{\bar{\Lambda}}(t',x)=0.\nonumber
\end{eqnarray}
That is, $\{f,g\}_{\bar{\Lambda}}$ is a $t$-dependent constant of
motion for $X$.
\end{proof}

\section{Polynomial Lie integrals for Lie--Hamilton systems}

Let us formally define and investigate a remarkable class of constants
of motion for Lie--Hamilton systems appearing in the literature  \cite{Ru10,Ma95}, the hereafter called {\it Lie integrals}, and a relevant generalization of them, 
the {\it polynomial Lie integrals}. We first prove that Lie integrals can be characterised by an Euler equation on a finite-dimensional real Lie algebra of functions, 
retrieving as a particular case a result given in \cite{Ru10}. 
Then, we show that Lie integrals form a finite-dimensional real Lie algebra and we devise several methods to determine them. Our
results can easily be extended to investigate certain quantum mechanical systems \cite{LT05}. Finally, we investigate polynomial Lie integrals
and the relevence of Casimir functions to derive them.

\begin{definition} Given a Lie--Hamilton system $X$ on $N$ possessing a Lie--Hamiltonian structure
$(N,\Lambda,h)$, a {\it Lie integral} of $X$ with respect to
$(N,\Lambda,h)$ is a constant of motion $f$ of $X$ such that $\{f_t\}_{t\in\mathbb{R}}\subset \ham$. In other words, given a basis $h_1,\ldots,h_r$
of the Lie algebra $(\ham,\{\cdot,\cdot\}_\Lambda)$, we have that $\bar Xf=0$ and
$f_t=\sum_{\alpha=1}^r f_\alpha(t)h_\alpha$ for every $t\in\mathbb{R}$ and certain $t$-dependent functions $f_1,\ldots,f_r$.
\end{definition}

The Lie integrals of a Lie--Hamilton system $X$ relative to
 a Lie--Hamiltonian structure $(N,\Lambda,h)$ are the solutions of  the equation
\[
0=\bar{X}f=\frac{\partial f}{\partial t}+Xf=\frac{\partial f}{\partial t}+\{f,h\}_{\bar\Lambda}\Longrightarrow \frac{\partial f}{\partial t}=\{h,f\}_{\bar\Lambda}.
\]
Since $f$ and $h$ can be understood as curves $t\mapsto f_t$ and $t\mapsto g_t$ within $\ham$, the above equation can be rewritten as
\begin{equation}\label{EqIn2}
\frac{{\rm d}f_t}{{\rm d} t}=\{h_t,f_t\}_{\Lambda},
\end{equation}
which can be thought of as an
Euler equation
on   the Lie algebra $(\mathcal{H}_\Lambda,\{\cdot,\cdot\}_\Lambda)$ \cite{FLV10}. Equations of this type quite
frequently appear in the literature such as in the Lewis--Riesenfeld method and
works concerning Lie--Hamilton systems \cite{Ru10,Ma95,LT05}.

\subsection{Algebraic structure of Lie integrals}
Previous to the description of methods to solve equation (\ref{EqIn2}), let us prove some results about the algebraic structure of its solutions and the equation itself.
\begin{proposition}
Given a  Lie--Hamilton system $X$ with a Lie--Hamiltonian structure $(N,\Lambda,h)$, the
space $\mathfrak{L}^\Lambda_h$ of Lie integrals relative to $(N,\Lambda,h)$
gives rise to a Lie algebra $(\mathfrak{L}^\Lambda_h,\{\cdot ,\cdot
\}_{\bar{\Lambda}})$ isomorphic to $(\ham,\{\cdot,\cdot\}_{\Lambda})$.
\end{proposition}
\begin{proof}
Since the Lie integrals of $X$ are the solutions of the system of ODEs (\ref{EqIn2}) on $\ham$, they span an $\mathbb{R}$-linear space of dimension $\dim\,\ham$. In
view of Proposition \ref{AutLieHam2}, the Poisson bracket $\{f,g\}_{\bar\Lambda}$ of two constants of motion $f,g$ for $X$ is another
constant of motion. If $f$ and $g$ are Lie integrals, the function $\{f,g\}_{\bar \Lambda}$ is then a new constant of motion that can additionally be considered as a curve $t\mapsto \{f_t,g_t\}_{\Lambda}$ taking values in $\ham$, i.e.~a new Lie integral.

Consider now the linear morphism   $\evo:f\in \mathfrak{L}^\Lambda_h\rightarrow f_0\in \ham$ relating every Lie integral to its value in $\ham$ at $t=0$. As every initial condition in $\ham$ is related to a single solution of (\ref{EqIn2}), we can relate every $v\in \ham$ with a unique Lie integral $f$ of $X$ such that $f_0=v$. Therefore, $\evo$ is an isomorphism. Indeed, it is a Lie algebra isomorphism since $(\{f,g\}_{\bar\Lambda})_0=\{f_0,g_0\}_\Lambda$ for every $f,g\in\mathfrak{L}^\Lambda_h$.
\end{proof}

\begin{proposition} Given a Lie--Hamilton system $X$ on $N$ possessing a Lie--Hamiltonian structure $(N,\Lambda,h)$,
then $\mathfrak{L}^\Lambda_h$ consists of $t$-independent constants of motion if and only if $\ham$ is Abelian.
\end{proposition}
 \begin{proof}
If $(\ham,\{\cdot,\cdot\}_\Lambda)$ is Abelian, then $\{f_t,h_t\}_{\Lambda}=0$ and the system (\ref{EqIn2}) reduces to
${{\rm d}f_t}/{{\rm d} t}=0$, whose solutions are of the form $f_t=g\in\ham$, i.e.~$\mathfrak{L}^\Lambda_h=\ham$.
  Conversely, if $\mathfrak{L}^\Lambda_h=\ham$, then every $g\in \ham$ is a solution of (\ref{EqIn2}) and $\{g,h_t\}_\Lambda=0$   $\forall t\in\mathbb{R}$. 
  Hence, every $g\in \ham$ commutes with the whole $\ham$, which becomes Abelian.
\end{proof}

\subsection{Reduction techniques for the determination of Lie integrals}
Let us turn to the study of (\ref{EqIn2}) through the methods of the theory
of Lie systems. A first approach to this topic can be found in \cite{FLV10,Ru10,EV05}. To start with,
we prove that system (\ref{EqIn2}) is a Lie system related to a Vessiot--Guldberg Lie algebra $V$ isomorphic to $ ({\ham}/Z( {\ham}),\{\cdot,\cdot\}_{\Lambda})$.

Let $\{h_1,\ldots, h_r\}$ be a basis for $ {\ham}$, we can write $h_t=\sum_{\alpha=1}^rb_\alpha(t)h_\alpha$ for certain $t$-dependent functions $b_1,\ldots,b_r$. In addition, system (\ref{EqIn2}) becomes
\[
\frac{{\rm d}f_t}{{\rm d} t}=\sum_{\alpha=1}^rb_\alpha(t)Y_\alpha(f_t),
\]
where $Y_\alpha(f_t)={\rm ad}_{h_\alpha}(f_t)$, with $\alpha=1,\ldots,r$ and ${\rm ad}:f\in {\ham}\mapsto {\rm ad}_f(g)\equiv\{f,g\}_\Lambda\in  {\ham}$ is the adjoint representation of $ {\ham}$. The vector fields $Y_1,\ldots, Y_r$ are the
fundamental vector fields of the adjoint action ${\rm Ad}: \grupo\times
 {\ham}\rightarrow  {\ham}$ of a Lie group $\grupo$ with Lie algebra
isomorphic to $ {\ham}$. Consequently, they span a finite-dimensional real
Lie algebra of vector fields isomorphic to ${\rm Im}\,{\rm ad}\simeq {\ham}/\ker {\rm
ad}= {\ham}/Z( {\ham}),$
where we write $Z( {\ham})$ for the {\it center} of $ ({\ham},\{\cdot,\cdot\}_\Lambda)$. So, (\ref{EqIn2}) is a Lie system related to a Vessiot--Guldberg Lie
algebra isomorphic to  $ ({\ham}/Z( {\ham}),\{\cdot,\cdot\}_\Lambda)$.

Instead of applying the standard approaches of the theory of Lie systems to (\ref{EqIn2}) (see for instance \cite{CGM00,CGM07,Dissertationes,CarRamGra}),
we hereafter derive a more straightforward new method based upon Lie--Hamilton systems that provides similar results.

If $Z( {\ham})$ is not trivial, its elements are Lie integrals of $X$. More generally, if $ ({\ham},\{\cdot,\cdot\}_\Lambda)$ is
 solvable, ${\ham}$ admits a flag of ideals \cite{Po86}, i.e.~there exists  a family $ {\hamb}_1,\ldots, {\hamb}_{r-1}$ of ideals of $ {\ham}$ satisfying that
\[
  {\ham} \equiv {\hamb}_0\supset  {\hamb}_1\supset\ldots\supset {\hamb}_{r-1}\supset
{\hamb}_{r}=0,\qquad \dim  {\hamb}_{k-1}/ {\hamb}_{k}=1, \qquad  k=1,\ldots,r .
\]
 Hence, we can construct an adapted basis
$\{h_1,\ldots,h_r\}$ of $ {\ham}$ such that $\{h_k,\ldots,h_r\}$ form a basis for
$ {\hamb}_{k-1}$ for $k=1,\ldots,r$. If $\alpha<\beta$, we have $ {\hamb}_\alpha\supset  {\hamb}_\beta$ and, as $ {\hamb}_\beta$ is an ideal, we see that $\{h_\alpha,h_\beta\}_\Lambda\in
 {\hamb}_{\beta}$. Consequently, if we write
$f_t=\sum_{\alpha=1}^rf_\alpha(t)h_\alpha$ and $ \{h_\alpha,h_\beta\}_\Lambda=\sum_{\gamma=1}^r c_{\alpha\beta\gamma} h_\gamma$, equation (\ref{EqIn2}) reads
\[
  \qquad \sum_{\alpha=1}^r\left(\frac{{\rm d}f_\alpha}{{\rm d} t}h_\alpha+\sum_{\beta=1}^rb_\beta
f_\alpha\{h_\alpha,h_\beta\}_\Lambda\right)=\sum_{\alpha=1}^r\left(\frac{{\rm d}f_\alpha}{{\rm d} t}+\sum_{\beta,\gamma=1}^rb_\beta
f_\gamma c_{\gamma\beta\alpha}\right)h_\alpha=0,
\]
where  $c_{\gamma\beta\alpha}=0$ for $\alpha<{\rm max}(\gamma,\beta)$. Therefore, we obtain the easily integrable system
\[
\frac{{\rm d}f_\alpha}{{\rm d} t}=
\sum_{\beta=1}^rT_{\alpha\beta}(t)f_\beta
,\quad \alpha=1,\ldots,r,\quad \Longrightarrow \quad \frac{{\rm d}f}{{\rm d} t}=
T(t)f,\qquad f\in \ham,
\]
where $T(t)$ is a lower triangular $r\times r$ matrix with entries $T_{\alpha\gamma}(t)=
-\sum_{\beta=1}^rb_\beta(t) c_{\gamma\beta \alpha}$ and $f=(f_1,\ldots,f_r)^T$ is considered as an $r\times 1$ matrix. 

The previous method can   be generalised for {\em any} Lie--Hamilton algebra. Let us use a {\it Levi decomposition} \cite{Po86} for $ ({\ham},\{\cdot,\cdot\}_\Lambda)$ by writing $ {\ham}=\left( {\hamb}_{s_1}\oplus\ldots\oplus {\hamb}_{s_p}\right)\oplus_s {\hamb}_\rr$, where    $ {\hamb}_{s_1}\oplus \ldots\oplus {\hamb}_{s_p}$ is a direct sum of simple subalgebras of
$ {\ham}$, the Lie subalgebra $ {\hamb}_\rr$ is its {\it radical}, i.e.~its maximal solvable ideal, and $\oplus_s$ stands for a semidirect sum. Let us also denote $\lambda_i=\dim  {\hamb}_{s_i}$ for $i=1,\ldots,p$ and $\mm=\dim  {\hamb}_{\rr}$.
Then, every element $f\in {\ham}$ can be written in a unique way as $f_\rr+f_{s_1}+\ldots+f_{s_p}$, where $f_\rr\in {\hamb}_\rr$ and $f_{s_i}\in {\hamb}_{s_i}$ for $i=1,\ldots,p$. Proceeding as before, we can choose
a basis $h_1,\ldots,h_\mm$ adapted to a flag decomposition of the solvable Lie algebra $ {\hamb}_\rr$ and extend it to a basis of $ \ham$ by choosing a basis $h^{(s_i)}_{1},\ldots,h^{(s_i)}_{\lambda_i}$ for each $ {\hamb}_{s_i}$. In this way, we   obtain that
\begin{equation}\label{RedEq}
\frac{{\rm d}f_\rho}{
{\rm d}t}=T(t)f_\rho+\sum_{i=1}^pB_i(t)f_{s_i},\qquad \frac{{\rm d}f_{s_j}}{{\rm d}t}=C_j(t)f_{s_j},\quad j=1,\ldots,p,
\end{equation}
where, as before, $T(t)$ is a lower triangular $\mm\times \mm$ matrix, 
while $B_i(t)$ and $C_i(t)$ are $\mm \times\lambda_i $ and $\lambda_i\times\lambda_i$ matrices
with $t$-dependent coefficients for $i=1,\ldots,p$. Consequently,
 the initial system (\ref{EqIn2}) reduces to working out the equations for $f_{s_1},\ldots,f_{s_p}$ in each $ {\hamb}_{s_i}$ separately. From that, $f_\rr$ can   be  determined through an affine system whose homogeneous associated system is easily integrable.

If $ {\ham}$ is a {\em simple} Lie algebra, (\ref{RedEq}) is just an expression in coordinates
of the system (\ref{EqIn2}), which does not provide any simplification. However, when  $ {\ham}$ is a
{\em  non-simple} one,  our techniques do simplify the calculation of Lie integrals. Although similar results can be achieved 
by using the procedure developed in \cite{CarRamGra} that, nevertheless,    requires the use of certain transformations and integrations of ODEs, which are here unnecessary.

\subsection{Polynomial Lie integrals and Casimir functions}

We here investigate a generalization of Lie integrals: the hereafter called {\it polynomial Lie integrals}.
Although we prove that these constants of motion can be determined by Lie integrals, we also show that 
their determination can be simpler in some cases. In particular, we can obtain polynomial Lie integrals algebraically by means of the Casimir functions related to the Lie algebra of Lie integrals.

\begin{definition} Let $X$ be a Lie--Hamilton system admitting a compatible Lie--Hamiltonian
structure $(N,\Lambda,h)$. A {\it polynomial Lie integral} for $X$ with respect to $(N,\Lambda,h)$
is a constant of motion $f$ for $X$ of the form $f_t=\sum_{\JJ\in M}\lambda_\JJ(t)h^\JJ$,
where the $\JJ$'s are $r$-multi-indexes, i.e.~sets $(i_1,\ldots,i_r)$ of nonnegative integers, the set $M$ is a finite family of multi-indexes, the $\lambda_\JJ(t)$ are certain $t$-dependent functions, and $h^\JJ=h_1^{i_1}\cdot\ldots\cdot h_r^{i_r}$ for a fixed basis $\{h_1,\ldots,h_r\}$ for $\ham$.
\end{definition}

The study of polynomial Lie integrals can be approached through the symmetric Lie algebra $S_\mathfrak{g}$,
where $\mathfrak{g}\simeq \ham$.

\begin{lemma}
\label{62} Every Lie algebra isomorphism $\phi:(\mathfrak{g},[\cdot,\cdot]_{\mathfrak{g}})\rightarrow
(\ham,\{\cdot,\cdot\}_\Lambda)$
can be extended in a unique way to a Poisson algebra morphism $D:(S_\mathfrak{g},\cdot,\{\cdot,\cdot\}_{S_\mathfrak{g}})\rightarrow (C^\infty(N),\cdot,\{\cdot,\cdot\}_\Lambda)$.
Indeed, if $\{v_1,\ldots,v_r\}$ is a basis for $\mathfrak{g}$, then
$
D(P(v_1,\ldots,v_r))=P(\phi(v_1),\ldots,\phi(v_r)),
$
for every polynomial $P\in S_\mathfrak{g}$.
\end{lemma}

The proof is addressed  in the Appendix. Recall that ``$\cdot$'' denotes the standard product of elements
of $S_\mathfrak{g}$ understood as polynomial functions on $S_\mathfrak{g}$. 
It is remarkable that $D$ does not need to be injective, which causes that $S_\mathfrak{g}$ 
is not in general isomorphic to the space of polynomials on the elements of a basis of $\ham$. For instance, 
consider the Lie algebra isomorphism $\phi:(\mathfrak{sl}(2,\mathbb{R}),[\cdot,\cdot]_{\mathfrak{sl}(2,\mathbb{R})})\rightarrow (\ham,\{\cdot,\cdot\})_\Lambda$, 
with $\{v_1,v_2,v_3\}$ being a basis of $\mathfrak{sl}(2,\mathbb{R})$, of the form  $\phi(v_1)=p^2$, $\phi(v_2)=xp$ and $\phi(v_3)=x^2$ 
and $\{\cdot,\cdot\}$ being the standard Poisson structure on ${\rm T}^*\mathbb{R}$. Then, $D(v_1v_3-v_2^2)=\phi(v_1)\phi(v_3)-\phi^2(v_2)=0$.  

The following notion allows us to simplify the statement and proofs of our results.
\begin{definition} Given a curve $P_t$ in $S_\mathfrak{g}$, its {\it degree}, $\degr(P_t)$, is the highest degree of the polynomials $\{P_t\}_{t\in\mathbb{R}}$. If there exists no finite highest degree, we say that $\degr(P_t)=\infty$.
\end{definition}

\begin{proposition}\label{Th:Lift} A function $f$ is a polynomial Lie integral for a Lie--Hamilton system $X$
with respect to the Lie--Hamiltonian  structure $(N,\Lambda,h)$ if and only
if  for every $t\in\mathbb{R}$ we have $f_t=D(P_t)$, where $D$ is the Poisson algebra morphism
$D:(S_\mathfrak{g},\cdot,\{\cdot,\cdot\}_{{S}_\mathfrak{g}})\rightarrow (C^\infty(N),\cdot,\{\cdot,\cdot\}_\Lambda)$ induced by $\phi:(\mathfrak{g},[\cdot,\cdot]_{\mathfrak{g}})\rightarrow
(\ham,\{\cdot,\cdot\}_\Lambda)$, and the curve $P_t$ is a solution of finite degree for
\begin{equation}\label{PolLie}
\frac{{\rm d}P}{{\rm d} t}+\{P, w_t\}_{S_\mathfrak{g}}=0,\qquad P\in S_\mathfrak{g},
\end{equation}
where $w_t$ stands for a curve in $\mathfrak{g}$ such that $D(w_t)=h_t$ for every $t\in\mathbb{R}$.
\end{proposition}
\begin{proof}
Let $P_t$ be a particular solution of (\ref{PolLie}). Since $D$ is a Poisson algebra morphism and  $h_t=D(w_t)$ for every $t\in\mathbb{R}$, we obtain by applying $D$ to (\ref{PolLie}) that
\begin{equation}\label{PolLie3}\nonumber
\frac{\partial D(P_t)}{\partial t}+\{D(P_t), D(w_t)\}_{\bar\Lambda}=\frac{\partial D(P_t)}{\partial t}+\{D(P_t), h_t\}_{\bar\Lambda}=0.
\end{equation}
So, $D(P_t)$ is a Lie integral for $X$. Note that this does not depend on the chosen curve $w_t$ satisfying $D(w_t)=h_t$.

Conversely, given a polynomial Lie integral $f$ for $X$, there exists a curve $P_t$ of finite degree
such that $D(P_t)=f_t$ for every $t\in\mathbb{R}$. Hence, we see that
\begin{eqnarray*}\label{PolLie4}
  D\left(\frac{{\rm d}P_t}{{\rm d} t}+\{P_t, w_t\}_{S_\mathfrak{g}}\right)=\frac{\partial D(P_t)}{\partial t}+\{D(P_t), D(w_t)\}_{\bar\Lambda}=0\Longrightarrow \frac{{\rm d}P_t}{{\rm d} t}+\{P_t, w_t\}_{S_\mathfrak{g}}=\xi_t,
\end{eqnarray*}
where $\xi_t$ is a curve in $\ker D$. As $\degr(dP_t/{\rm d} t)$ and $\degr(\{P_t,w_t\}_{{S_\mathfrak{g}}})$ are at most $\degr(P_t)$, then $\degr(\xi_t)\leq \degr(P_t)$.
Next, consider the equation
\[\label{PolLie5}
\frac{{\rm d}\eta}{{\rm d} t}+\{\eta, w_t\}_{S_\mathfrak{g}}=\xi_t,\qquad \degr(\eta)\leq \degr(P)\quad\,\, {\rm and} \,\,\quad\eta\subset \ker D.
\]
Note that this equation is well defined. Indeed, since $\degr(\eta)\leq \degr(P)$  and $\degr(w_t)\leq 1$ for every $t\in\mathbb{R}$, then $\degr (\{\eta, w_t\}_{S_\mathfrak{g}})\leq\degr(P)$ for all $t\in\mathbb{R}$. In addition, as $D(\eta_t)\subset \ker D$, then $\{\eta, w_t\}_{S_\mathfrak{g}}\in \ker D$. Then, the above equation can be restricted to the finite-dimensional space of elements of $\ker D$ with degree at most $\degr(P_t)$. Given a particular
solution $\eta_t$ of this equation, which exists for the Theorem of existence and uniqueness, we have that $P_t-\eta_t$ is a solution of (\ref{PolLie}) projecting
into $f_t$.
\end{proof}

\begin{proposition}\label{Pr:De}
Every polynomial Lie integral $f$ of a Lie--Hamilton system
$X$ admitting a Lie--Hamiltonian structure $(N,\Lambda,h)$ can be brought into the form
$f=\sum_{\JJ\in M}c_\JJ l^\JJ,$
where $M$ is a finite set of multi-indexes, the $c_\JJ$'s are certain real constants, and
$l^\JJ=f_1^{i_1}\cdot\ldots\cdot f_r^{i_r}$, with $f_1,\ldots,f_r$ being a basis of the space
$\mathfrak{L}^\Lambda_h$.
\end{proposition}
\begin{proof}
From Proposition \ref{Th:Lift}, we have that $f_t=D(P_t)$ for a solution $P_t$ of finite degree $p$ for (\ref{PolLie}). So, it is a solution of the restriction of this system
to $\mathbb{P}(p,\mathfrak{g})$, i.e.~the elements of $S_{\mathfrak{g}}$ of degree at most $p$. Given the isomorphism $\phi:\mathfrak{g}\rightarrow \ham$, define $\phi^{-1}(f_j)$, with $j=1,\ldots,r$, to be the curve in $\mathfrak{g}$ of the form $t\mapsto \phi^{-1}(f_j)_t$. Note that $v_1\equiv\phi^{-1}(f_1)_0,\ldots, v_r\equiv\phi^{-1}(f_r)_0$ form a basis of $\mathfrak{g}$. Hence, their polynomials up to order $p$ span a basis for $\mathbb{P}(p,\mathfrak{g})$ and we can write $P_0=\sum_{\JJ\in M}c_\JJ v^\JJ$, where $v^\JJ=v_1^{i_1}\cdot\ldots\cdot v_r^{i_r}$. As $P'_t=\sum_{I\in M}c_I[\phi^{-1}(f_1)]^{i_1}_t\cdots [\phi^{-1}(f_r)]^{i_r}_t$ and $P_t$ are solutions with the same initial condition of the restriction of (\ref{PolLie}) to $\mathbb{P}(p,\mathfrak{g})$, they must be the same in virtue of the Theorem of existence and uniqueness of systems of differential equations. Applying $D$, we obtain that $f_t=D(P_t)=D(\sum_{\JJ\in M}c_\JJ [\phi^{-1}(f_1)_t]^{i_1}\cdots [\phi^{-1}(f_r)_t]^{i_r})=\sum_{\JJ\in 
M}c_\JJ l^\JJ_t$.
\end{proof}

\begin{corollary}\label{ST}
Let $X$ be a Lie--Hamilton system that possesses a Lie--Hamiltonian structure $(N,\Lambda,h)$ inducing a Lie algebra isomorphism $\phi:\mathfrak{g}\rightarrow\ham$ and a Poisson algebra morphism $D:S_\mathfrak{g}\rightarrow C^\infty(N)$. The function $\FF=D(\casS)$, where $\casS$ is a Casimir element of $S_\mathfrak{g}$, is a
$t$-independent constant of motion of $X$. If $\casU$ is
a Casimir element of $U_\mathfrak{g}$, then $F=D(\lambda^{-1}(\casU))$ is
$t$-independent constant of motion for $X$.
\end{corollary}

Note that if $C$ is a constant of motion for $X$, it is also so for any other $X'$ whose $V^{X'}\subset V^X$.    
From  Proposition \ref{Pr:De} and Corollary \ref{ST}, it follows that $F=D(\casS)=\sum_{\JJ\in M}c_\JJ l^\JJ$. Therefore,
the knowledge of Casimir elements provides not only constants of motion for Lie--Hamilton systems, but also information
about the Lie integrals of the system.  

As Casimir functions are
known for many Lie
algebras, we can use them to derive constants of motion for the corresponding Lie--Hamilton systems algebraically instead of applying the usual procedure, i.e.~by solving a system of  PDEs or ODEs.

In particular, Casimir functions for (semi)simple Lie algebras of arbitrary dimension are known \cite{GO64,PP68}. The same is true for the so-called ``quasi-simple" Lie algebras, which can be obtained from simple Lie algebras through contraction techniques \cite{Vulpi}. Moreover, the Casimir invariants (Casimir elements of the Poisson algebra $(C^\infty(\mathfrak{g}^*),\{\cdot,\cdot\})$) for all real Lie algebras with dimension $d\leq 5$ were given in \cite{invariantsWint} (recall that the Casimir invariants for some of the solvable cases are not polynomial, i.e. they cannot be considered as elements of $S_\mathfrak{g}$), and the literature dealing with Casimir invariants for solvable and nilpotent Lie algebras is not scarce (see, e.g. \cite{campoamor,campoamor1,campoamor2}).

\section{Superposition rules from Poisson coalgebras}

We here prove that each Lie--Hamiltonian structure $(N,\Lambda,h)$ for a Lie--Hamilton system $X$ gives rise  in a natural way    to a Poisson coalgebra  $(S_\mathfrak{g},\cdot,\{\cdot,\cdot\}_{S_\mathfrak{g}},\Delta)$ where $\mathfrak{g}
\simeq(\ham,\{\cdot,\cdot\}_\Lambda)$.
This allows us to use the coproduct of this coalgebra
to construct new Lie--Hamiltonian structures
for all the diagonal prolongations of $X$ and to derive algebraically $t$-independent constants of  motion for 
such diagonal prolongations. In turn, these constants can further be employed
to obtain a superposition rule for the initial system. Our findings, which are only applicable to 
``primitive" Poisson coalgebras,
are rigorous demonstrations and generalisations of previous achievements established in \cite{BBHMR09,BR,BCR96}.

\subsection{Poisson coalgebras for Lie--Hamiltonian structures}

\begin{lemma}\label{PAM} Given a Lie--Hamilton system $X$ with a Lie--Hamiltonian structure $(N,\Lambda,h)$,
the space $(S_\mathfrak{g},\cdot,\{\cdot,\cdot\}_{S_\mathfrak{g}},\Delta)$, with $\mathfrak{g}
\simeq(\ham,\{\cdot,\cdot\}_\Lambda)$, is a Poisson coalgebra with a coproduct ${\Delta}:
S_\mathfrak{g}\rightarrow
S_\mathfrak{g}\otimes S_\mathfrak{g}$ satisfying
\begin{equation}\label{Con}
{\Delta}(v)=v\otimes 1+1\otimes v,\qquad \forall  v\in\mathfrak {g}\subset S_\mathfrak{g}.
\end{equation}
\end{lemma}

\begin{proof}
We know that $(S_\mathfrak{g},\cdot,\{\cdot,\cdot\}_{S_\mathfrak{g}})$ and $(S_\mathfrak{g}\otimes S_\mathfrak{g},\cdot_{S_\mathfrak{g}\otimes S_\mathfrak{g} },\{\cdot,\cdot\}_{S_\mathfrak{g}\otimes S_\mathfrak{g} })$ are Poisson algebras. The coassociativity property for the coproduct map (\ref{Con}) is straightforward. Therefore, let us prove
that there exists a Poisson algebra morphism $\Delta:(S_\mathfrak{g},\cdot,\{\cdot,\cdot\}_{S_\mathfrak{g}})
\rightarrow (S_\mathfrak{g}\otimes S_\mathfrak{g},\cdot_{S_\mathfrak{g}\otimes S_\mathfrak{g}},\{\cdot,\cdot\}_{S_\mathfrak{g}\otimes S_\mathfrak{g}})$
satisfying (\ref{Con}), which turns $(S_\mathfrak{g},\cdot,\{\cdot,\cdot\}_{S_\mathfrak{g}},\Delta)$ into a Poisson coalgebra.

The elements of $S_\mathfrak{g}$ of the form $v^\JJ\equiv v_1^{i_1}\cdot\ldots\cdot v_r^{i_r}$, where the $\JJ$'s are $r$-multi-index with $r=\dim\,\mathfrak{g}$, form a basis for $S_\mathfrak{g}$ (considered as a linear space).
 Then, every $P\in S_\mathfrak{g}$ can be written in a unique way as $
P=\sum_{\JJ\in M}\lambda_\JJ v^\JJ$, 
where the $\lambda_\JJ$ are real constants and $\JJ$ runs all the elements of a finite set $M$. In view of this, 
an $\mathbb{R}$-algebra morphism  $\Delta:S_\mathfrak{g}\rightarrow
S_\mathfrak{g}\otimes S_\mathfrak{g}$ is determined by the image of the elements of a basis for $\mathfrak{g}$, i.e.
\begin{equation}\label{dec}
\Delta(P)=\sum_\JJ\lambda_\JJ\Delta(v^\JJ)=\sum_\JJ\lambda_\JJ\Delta(v^{i_1}_1)\cdot\ldots\cdot\Delta(v^{i_r}_r).
\end{equation}
Therefore, two $\mathbb{R}$-algebra morphisms that coincide on the elements on $\mathfrak{g}$ are necessarily the same.
 Hence, if
there exists such  a morphism, it is unique. Let us prove that there exists an $\mathbb{R}$-algebra morphism
$\Delta$ satisfying (\ref{Con}).

From (\ref{dec}), we easily see that $\Delta$ is $\mathbb{R}$-linear. Let us also prove that $\Delta(PQ)=\Delta(P)\Delta(Q)$
for every $P,Q\in S_\mathfrak{g}$, which shows that $\Delta$ is an $\mathbb{R}$-algebra morphism. If we write $Q=\sum_{\KK\in M}\lambda_\KK v^\KK$, we obtain that
\[
  \Delta(PQ)=\sum_\LL\left(
\sum_{\JJ+\KK=\LL}\lambda_\JJ\lambda_\KK\right)\Delta(v^\LL)=\sum_\JJ\lambda_\JJ\Delta(v^\JJ)\sum_\KK\lambda_\KK\Delta(v^\KK )=\Delta(P)\Delta(Q).\nonumber
\]

Finally, we show that  $\Delta$ is also a Poisson morphism. By linearity, this reduces to proving that $\Delta(\{v^\JJ,v^\KK\}_{S_\mathfrak{g}})=\{\Delta (v^\JJ),\Delta( v^\KK)\}_{S_\mathfrak{g}\otimes S_\mathfrak{g}}$.  If $|\JJ|=0$ or $|\KK|=0$, this result is immediate as the Poisson bracket involving a constant is zero. For the remaining cases and starting by $|\JJ|+|\KK|=2$, we have that ${\Delta}(\{v_\alpha,v_\beta\}_{S_\mathfrak{g}})=\{{\Delta}( v_\alpha),{\Delta}(
v_\beta)\}_{S_\mathfrak{g}\otimes S_\mathfrak{g}}$, $\forall \alpha,\beta=1,\ldots,r.$ Proceeding by induction, we prove that this holds for every value of $|\JJ|+|\KK|$; by writing $v^\JJ= v^{\bar \JJ}v_\gamma^{i_\gamma}$ and using induction hypothesis, we get
\begin{eqnarray}
  \Delta\left(\{v^\JJ,v^\KK\}_{S_\mathfrak{g}}\right)&=\Delta\left (\{v^{\bar \JJ}v_\gamma^{i_\gamma},v^\KK\}_{S_\mathfrak{g}}\right)=\Delta\left(\{v^{\bar \JJ},v^\KK\}_{S_\mathfrak{g}}v_\gamma^{i_\gamma}+v^{\bar \JJ}\{v_\gamma^{i_\gamma},v^\KK\}_{S_\mathfrak{g}}\right)\cr
  &= \left\{\Delta (v^{\bar \JJ}),\Delta(v^\KK)\right\}_{S_\mathfrak{g}\otimes S_\mathfrak{g}}\!\! \Delta( v_\gamma^{i_\gamma})+\Delta (v^{\bar \JJ})\left\{\Delta (v_\gamma^{i_\gamma}),\Delta (v^\KK)\right\}_{S_\mathfrak{g}\otimes S_\mathfrak{g}} \cr
  &=\left\{\Delta (v^{\bar \JJ})\Delta (v_\gamma^{i_\gamma}),\Delta (v^\KK)\right\}_{S_\mathfrak{g}\otimes S_\mathfrak{g}}=\left\{\Delta (v^\JJ),\Delta (v^\KK)\right\}_{S_\mathfrak{g}\otimes S_\mathfrak{g}}.\nonumber
 \end{eqnarray}\end{proof}

The coproduct defined in the previous lemma gives rise to a new Poisson algebra morphism as stated in the
following immediate lemma.
\begin{lemma}
\label{L72}
The map $\Delta^{(m)}:(S_\mathfrak{g},\cdot,\{\cdot,\cdot \}_{S_\mathfrak{g}})\rightarrow (S_\mathfrak{g}^{(m)},\cdot_{S_\mathfrak{g}^{(m)}},\{\cdot,\cdot\}_{S_\mathfrak{g}^{(m)}})$ , with $m>1$, defined by recursion following (\ref{copr}) with $\Delta^{(2)}=\Delta$ given by (\ref{Con}), is a Poisson algebra morphism. 

\end{lemma}

The injection $\iota:\mathfrak{g}\rightarrow  \ham\subset
C^\infty(N)$ is a Lie algebra morphism that can be extended to a Poisson algebra morphism $D:S_\mathfrak{g}\rightarrow C^\infty(N)$
of the form
$D(P(v_1,\ldots,v_r))=P(\iota(v_1),\ldots,\iota(v_r))$. Recall that this map need not to be
injective.

\begin{lemma}
\label{TuTu}
The Lie algebra morphism $\mathfrak{g}\hookrightarrow C^\infty(N)$ gives rise to a family of
Poisson algebra morphisms $D^{(m)}:S^{(m)}_\mathfrak{g} \hookrightarrow C^\infty(N)^{(m)} \subset
C^\infty(\, {N^m}\,)$ satisfying, for all $v_1,\ldots,v_m\in \mathfrak{g}\subset S_\mathfrak{g}$,  that
\begin{equation}\label{induced}
  \left[D^{(m)}(v_{1}\otimes\ldots\otimes
v_{m})\right]\!(x_{(1)},\ldots,x_{(m)})\!=\![D(v_{1})](x_{(1)})\cdot\ldots\cdot
[D(v_{m})](x_{(m)}),
\end{equation}
where $x_{(i)}$ is a point of the manifold $N$ placed in the $i$-position within the product $N\times\ldots\times N\,\equiv N^m$.

\end{lemma}

\subsection{Constants of   motion from Poisson coalgebras}

From the above results, we can easily demonstrate the following statement which
shows that the diagonal prolongations of a Lie--Hamilton system $X$ are also
Lie--Hamilton ones admitting a structure induced by that of $X$.

\begin{proposition}\label{MT1} If $X$ is a Lie--Hamilton system on $N$ with a Lie--Hamiltonian structure $(N,\Lambda,h)$,
then the diagonal
prolongation $\widetilde X$ to each $N^{m+1}$ is also a Lie--Hamilton system endowed with a Lie--Hamiltonian structure $(N^{m+1},
\Lambda^{m+1},\widetilde h)$ given by
$\Lambda^{m+1}(x_{(0)},\ldots,x_{(m)})=\sum_{a=0}^m\Lambda(x_{(a)})$, where we make use of the vector bundle
isomorphism ${\rm T}N^{m+1}\simeq {\rm T}N\oplus\cdots\oplus{\rm T}N$,
and $\widetilde h_t=D^{(m+1)}({\Delta}^{(m+1)}(h_t))$, where $D^{(m+1)}$ is the Poisson algebra morphism (\ref{induced}) induced
by the Lie algebra morphism $\mathfrak{g}\hookrightarrow \ham\subset C^\infty(N)$.
\end{proposition}

The above results enable us to prove the following theorem that provides  a method
to obtain $t$-independent constants of motion for the diagonal
prolongations of a Lie--Hamilton system. From this theorem, one
may obtain superposition rules for Lie--Hamilton
systems in an algebraic way.  Additionally,  this theorem is a generalization,  only valid in the case of primitive coproduct maps, of the integrability theorem for coalgebra symmetric systems given in \cite{BR}.

\begin{theorem}\label{MT}
 If $X$ is a Lie--Hamilton system with a
Lie--Hamiltonian structure $(N,\Lambda,h)$ and $\casS$ is a Casimir element of $(S_\mathfrak{g}
,\cdot,\{,\}_{S_\mathfrak{g}})$, where $\mathfrak{g}\simeq \ham$, then:\\
(i)
The functions defined as
\begin{equation}\label{invA}
\FF^{(k)} = D^{(k)}(\Delta^{(k)}({\casS})) ,   \qquad k=2,\ldots,m,
\end{equation}
are  $t$-independent constants of motion  for the diagonal prolongation
$\widetilde X$ to $N^m$. Furthermore, if all the $\FF^{(k)}$ are non-constant functions,  they form a set of  $(m-1)$      functionally independent functions  in involution.
\\
(ii) The functions given by
\begin{equation}\label{invB}
\FF_{ij}^{(k)}=S_{ij} ( \FF^{(k)}   ) , \qquad 1\le  i<j\le  k,\qquad k=2,\ldots,m,
\end{equation}
where $S_{ij}$ is the permutation of variables $x_{(i)}\leftrightarrow
x_{(j)}$, are $t$-independent constants of motion  for the diagonal prolongation
$\widetilde X$ to $N^m$.
\end{theorem}
  \begin{proof}
Every $P\in S^{(j)}_\mathfrak{g}$ can naturally be considered
as an element
$P\otimes \stackrel{(k-j)-{\rm times}}{\overbrace{ 1\otimes\ldots\otimes
1}}\,\, \in
S^{(k)}_\mathfrak{g}.
$
Since $j\le k$, we have that $\{\Delta^{(j)}(\bar v),\Delta^{(k)}( v)\}_{S^{(k)}_\mathfrak{g}}=
\{\Delta^{(j)}(\bar v),\Delta^{(j)}( v)\}_{S^{(j)}_\mathfrak{g}}$, $\forall \bar v, v\in\mathfrak{g}$. So,
\[
\{\Delta^{(j)}({\casS}),\Delta^{(k)}(v)\}_{S^{(k)}_\mathfrak{g}}=
\{\Delta^{(j)}({\casS}),\Delta^{(j)}(v)\}_{S^{(j)}_\mathfrak{g}}=\Delta^{(j)}\left(\{{\casS},v\}_{S^{(j)}_{ {\mathfrak{g}}  }}\right)=0.
\]
Hence, by using that every function $f\in C^\infty(N^j)$ can be understood as a function $\pi^*f\in C^\infty(N^k)$, being
$\pi:N^j\times N^{k-j}\rightarrow N^j$ the projection onto the first factor, and by applying the Poisson algebra morphisms   introduced in Lemma~\ref{TuTu} we get
\[
\left \{D^{(j)}(\Delta^{(j)}({\casS})),D^{(k)}(\Delta^{(k)}(v))\right\}_{\Lambda^k}
=\left \{\FF^{(j)} ,D^{(k)}(\Delta^{(k)}(v))\right\}_{\Lambda^k}=0 , \,\, \forall v\in\mathfrak{g},
\]
 which leads to $ \left\{\FF^{(j)},\FF^{(k)} \right\}_{\Lambda^k}=0$, that is,
the functions  (\ref{invA}) are in involution. By construction (see Lemma~\ref{L72}), if these are non-constant, then they are functionally independent functions since $\FF^{(j)}$ lives in  $N^{j}$, meanwhile  $\FF^{(k)}$ is defined on  $N^{k}$.

Let us prove now that all the functions (\ref{invA}) and (\ref{invB}) are   $t$-independent constants of motion for $\widetilde X$.
Using that $\ham\simeq \mathfrak{g}$ and
$X_t=-\widehat\Lambda\circ d h_t$, we see that $X$ can be brought into the form $X_t=-\widehat\Lambda\circ d\circ
D(v_t)$ for a unique curve $t\rightarrow v_t$ in $\mathfrak{g}$. From this and Proposition \ref{MT1}, it follows
\[
  \widetilde X_t=-\Lambda^m\circ d
D^{(m)}(\Delta^{(m)}(v_t) )\Longrightarrow
\widetilde X_t
( \FF^{(k)})=\left\{D^{(k)}(\Delta^{(k)}({\casS})),D^{(m)}(\Delta^{(m)}(v_t))\right\}_{\Lambda^m}=0.
\]
Then,  $\FF^{(k)}$ is a common
first-integral for every $\widetilde X_t$. Finally, consider the permutation operators $S_{ij}$, with $1\le i <j \le k$ for $k=2,\dots,m$. Note that
\begin{eqnarray*}
0=S_{ij} \left\{\FF^{(k)},D^{(m)}(\Delta^{(m)}(v_t))\right\}_{\Lambda^m} =  \left\{S_{ij}  ( \FF^{(k)} ),S_{ij} \left(D^{(m)}(\Delta^{(m)}(v_t)) \right)\right\}_{\Lambda^m} \\=\left\{ F_{ij}^{(k) }, D^{(m)}(\Delta^{(m)}(v_t))  \right\}_{\Lambda^m} =\widetilde X_t   (F_{ij}^{(k) }  ) .
\end{eqnarray*}
Consequently, the functions $F_{ij}^{(k) }$ are $t$-independent constants of motion for $\widetilde X$.
  \end{proof}

Note that the ``omitted" case with $k=1$ in the set of constants (\ref{invA}) is, precisely,  the one provided by Corollary~\ref{ST} as $\FF^{(1)}\equiv \FF=D({\casS})$. Depending on the system $X$, or more specifically, on the associated $\ham$, the function $\FF$ can be 
either a useless trivial constant or a relevant function. It is also worth noting that  constants (\ref{invB}) need not be functionally independent, but we can always  choose those fulfilling such a property. Finally, observe that if $X'$ is such that $V^{X'}\subset V^X$, then the functions (\ref{invA}) and (\ref{invB}) are also constants of motion for the diagonal prolongation of $X'$ to $N^m$.

\section{Applications}
Let us illustrate our methods through the
analysis of systems of physical and mathematical relevance. We firstly use the techniques described in Section 5.3 to easily retrive the so-called Lewis--Riesenfeld invariant
for Ermakov systems \cite{LA08} and the $t$-independent constant of motion for a  system of  Riccati equations employed to construct
the superposition rule for Riccati equations
\cite{CGM07}. Secondly, we use the coalgebra approach developed in Section 6 to deduce superposition rules for
Kummer--Schwarz equations \cite{CGL12,Be97},  Smorodinsky--Winternitz systems~\cite{WSUF65,WSUF67} with time-dependent frequency and also for a classical system with trigonometric nonlinearities,
whose integrability has recently drawn some attention \cite{ADR11,AW05}. As we avoid the integration of PDEs or ODEs used in standard methods, our techniques are simpler than previous ones.

\subsection{Classical Ermakov systems}

Let us consider the classical Ermakov system \cite{Dissertationes}:
\[
\left\{\begin{array}{rl}
\frac{{\rm d}^2x}{{\rm d} t^2}&=-\omega^2(t)x+\frac{b}{x^3},\\
\frac{{\rm d}^2y}{{\rm d} t^2}&=-\omega^2(t)y,
\end{array}\right.
\]
with a non-constant $t$-dependent frequency $\omega(t)$, being $b$  a real constant. This system appears in a number of
applications related to problems in quantum and classical mechanics \cite{LA08}.
By writting this system as a first-order one
\begin{equation}\label{ClasErm}
\left\{\begin{array}{rlrl}
\frac{{\rm d} x}{{\rm d} t}&=v_x,\qquad &\frac{{\rm d}v_x}{{\rm d} t}&=-\omega^2(t)x+\frac{b}{x^3},\\
\frac{{\rm d}y}{{\rm d} t}&=v_y,\qquad &\frac{{\rm d}v_y}{{\rm d} t}&=-\omega^2(t)y,
\end{array}
\right.
\end{equation}
we can apply the theory of Lie systems. Indeed, this is a Lie system related to a Vessiot--Guldberg Lie
algebra $V$ isomorphic to $\mathfrak{sl}(2,\mathbb{R})$  \cite{SIGMA}. In fact, system
(\ref{ClasErm}) describes the integral curves of the $t$-dependent vector field $X=X_3+\omega^2(t)X_1$,
where the vector fields
\begin{eqnarray*}
  X_1\!=\!-x\frac{\partial}{\partial v_x}\!-\!y\frac{\partial}{\partial v_y}, \
X_2\!=\!\frac12 \left(\!{v_x}\frac{\partial}{\partial v_x}\!+\!{v_y}\frac{\partial}{\partial v_y}\!-\!x\frac{\partial}{\partial x}\!-\!y\frac{\partial}{\partial y}\!\right),  
X_3=v_x\frac{\partial}{\partial x}\!+\!
v_y\frac{\partial}{\partial y}\!+\!\frac{b}{x^3}\frac{\partial}{\partial v_x},
\end{eqnarray*}
satisfy the commutation relations
\begin{equation}
\label{sl2R}
[X_1,X_2]=X_1,\qquad [X_1,X_3]= 2X_2,\qquad [X_2,X_3]=X_3.
\end{equation}

As a first new result we show that this is a Lie--Hamilton system.  The vector fields are Hamiltonian
with respect to the  Poisson bivector $\Lambda=\partial/\partial x\wedge \partial/\partial v_x+
\partial/\partial y\wedge \partial/\partial v_y$ provided that  $X_\alpha=- \widehat\Lambda(d h_\alpha)$ for $\alpha=1,2,3$. Thus,
 we find the  following Hamiltonian functions which form a basis for $(\ham,\{\cdot,\cdot\}_\Lambda)\simeq (\mathfrak{sl}(2,\mathbb{R}),[\cdot,\cdot])$:
\[
h_1=\frac 12(x^2+y^2),\qquad h_2=-\frac 12 (xv_x+yv_y),\qquad h_3=\frac 12\left(v_x^2+v_y^2+\frac{b}{x^2}\right),
\]
 as they fulfil
\begin{equation}
\label{sl2Rh}
\{h_1,h_2\}=-h_1,\qquad \{h_1,h_3\}=-2h_2,\qquad \{h_2,h_3\}=-h_3.
\end{equation}

Since $X=X_3+\omega^2(t)X_1$ and $\omega(t)$ is not a constant, every $t$-independent constant of motion $f$ for $X$ is a common first-integral for $X_1,X_2,X_3$. 
Instead of searching an $f$ by solving the system of
PDEs given by $X_1f=X_2f=X_3f=0$,
we use Corollary~\ref{ST}. This easily provides such a first integral  through the Casimir element  of the symmetric algebra of
${\mathfrak{sl}(2,\mathbb{R})}$. Explicitly,
 given a basis $\{\ve_1,\ve_2,\ve_3\}$
for $\mathfrak{sl}(2,\mathbb{R})$ satisfying
\begin{equation}
\label{sl2Rve}
[\ve_1,\ve_2]=-\ve_1,\qquad  [\ve_1,\ve_3]=-2\ve_2,\qquad [\ve_2,\ve_3]=-\ve_3,
\end{equation}
the Casimir element
 of $\mathfrak{sl}(2,\mathbb{R})$ reads $\casU=\frac 12 (\ve_1\widetilde\otimes \ve_3+\ve_3\widetilde\otimes \ve_1)-\ve_2\widetilde\otimes \ve_2\in  U_{\mathfrak{sl}(2,\mathbb{R})}$. Then, the inverse of symmetrizer
morphism (\ref{symmap}), $\lambda^{-1}: U_{\mathfrak{sl}(2,\mathbb{R})} \rightarrow S_{\mathfrak{sl}(2,\mathbb{R})}
$, gives rise to  the Casimir element of $S_{\mathfrak{sl}(2,\mathbb{R})}$:
\begin{equation}
\label{casSL2}
\casS=\lambda^{-1}(\casU)=\ve_1\ve_3-\ve^2_2 .
 \end{equation}
 According to Lemma \ref{62} we consider the  Poisson algebra morphism $D$  induced by the isomorphism $\phi:  {\mathfrak{sl}(2,\mathbb{R})} \to \ham$ defined by $\phi(\ve_\alpha)= h_\alpha$ for $\alpha=1,2,3$. 
 Subsequently, via Corollary~\ref{ST}, we obtain
\[
  F=D(\casS)= \phi(\ve_1)\phi(\ve_3)-\phi^2(\ve_2)
=
h_1h_3-h_2^2=(v_y x-v_x y)^2+b\left(1+\frac{y^2}{x^2}\right) .
\]
 In this way, we recover,
 up to an additive and multiplicative constant, the well-known Lewis--Riesenfeld
invariant \cite{LA08}. Note that when $\omega(t)$ is a constant, then $V^X\subset V$ and the function $F$ is also a constant of motion
for $X$ (\ref{ClasErm}).

\subsection{Riccati equations}

Let us turn to the system of Riccati equations  on $\mathcal{O}=\{(x_1,x_2,x_3,x_4)\,|\, x_i\neq x_j,i\neq j=1,\ldots,4\}\subset \mathbb{R}^4,$ given by
\begin{equation}\label{CoupledRic}
\frac{{\rm d} x_i}{{\rm d} t}=a_0(t)+a_1(t)x_i+a_2(t)x_i^2,\qquad i=1,\ldots,4,\\
\end{equation}
where $a_0(t),a_1(t),a_2(t)$ are arbitrary $t$-dependent functions.
The knowledge of a non-constant $t$-independent constant of motion for any system of this type
leads to obtaining a superposition rule for Riccati equations \cite{CGM07}.
Usually, this requires the integration of a system of PDEs~\cite{CGM07} or ODEs~\cite{PW}.
As in the previous subsection, we obtain such a $t$-independent constant of motion through algebraic methods by showing that (\ref{CoupledRic})
is a Lie--Hamilton system with a given Lie--Hamiltonian structure and obtaining an associated polynomial Lie integral.

Observe that (\ref{CoupledRic}) is  a Lie system related to
a $t$-dependent vector field $X=a_0(t)X_1+a_1(t)X_2+a_2(t)X_3$, where
\[
X_1=\sum_{i=1}^4\frac{\partial}{\partial x_i},\qquad X_2=\sum_{i=1}^4x_i\frac{\partial}{\partial x_i},\qquad X_3=\sum_{i=1}^4x_i^2\frac{\partial}{\partial x_i}
\]
span a Vessiot--Guldberg Lie algebra $V$ for (\ref{CoupledRic}) isomorphic to $\mathfrak{sl}(2,\mathbb{R})$ satisfying the same commutation relations (\ref{sl2R}). For simplicity, we assume $V^X=V$. Nevertheless, our final results are valid for any other case.

To show that (\ref{CoupledRic}) is a Lie--Hamilton system for arbitrary functions $a_0(t)$, $a_1(t)$, $a_2(t)$, we need to search for a symplectic form $\omega$ such that $V$ consists of Hamiltonian vector fields. By impossing $\mathcal{L}_{X_\alpha}\omega=0$, for $\alpha=1,2,3$, we obtain the 2-form
\[
\omega=\frac{{\rm d} x_1\wedge {\rm d} x_2}{(x_1-x_2)^2}+\frac{{\rm d} x_3\wedge {\rm d} x_4}{(x_3-x_4)^2},
\]
which is closed and non-degenerate on $\mathcal{O}$. Now, observe that
$\iota_{X_\alpha}\omega={\rm d}h_\alpha$, with $\alpha=1,2,3$ and
\[
  h_1=\frac{1}{x_1-x_2}+\frac{1}{x_3-x_4},\quad
h_2=\frac 12\left(\frac{x_1+x_2}{x_1-x_2}+\frac{x_3+x_4}{x_3-x_4}\right),\quad
h_3=\frac{x_1 x_2}{x_1-x_2}+\frac{x_3 x_4}{x_3-x_4}.
\]
So, $h_1,h_2$ and $h_3$ are Hamiltonian functions for $X_1$, $X_2$ and $X_3$, correspondingly. Using the Poisson bracket $\{\cdot,\cdot\}_\omega$ induced by $\omega$, we obtain that $h_1,h_2$ and $h_3$ satisfy the commutation relations (\ref{sl2Rh}), 
and $(\langle h_1,h_2,h_3\rangle,\{\cdot,\cdot\}_\omega) \simeq \mathfrak{sl}(2,\mathbb{R})$.
Next, we again express  $\mathfrak{sl}(2,\mathbb{R})$ in the   basis $\{\ve_1,\ve_2,\ve_3\}$  with Lie brackets   (\ref{sl2Rve}) and Casimir function (\ref{casSL2}), 
and we consider the Poisson algebra morphism $D: S_{{\mathfrak{sl}(2,\mathbb{R})}}\to   C^\infty ({\cal O})$  given by the isomorphism $\phi(\ve_\alpha)=\,h_\alpha$ for $\alpha=1,2,3$.
As $(\mathcal{O},\{\cdot,\cdot\}_\omega,h_t=a_0(t)h_1+a_1(t)h_2+a_2(t)h_3)$ is a Lie--Hamiltonian structure for $X$ and applying Corollary~\ref{ST}, we obtain  
the $t$-independent constant of motion for $X$:
\[
F=D(\casS)= h_1h_3-h_2^2=\frac{(x_1-x_4)(x_2-x_3)}{(x_1-x_2)(x_3-x_4)}.
\]
As in the previous example, if $V^X\subset V$, then $F$ is also a constant of motion for $X$. It is worth noting that $F$ is the
known constant of motion obtained for deriving a superposition rule for Riccati equations
\cite{PW,CGM07}, which is here deduced through a simple algebraic calculation.

It is also interesting that $V$ also becomes  a Lie algebra of Hamiltonian vector fields with respect to a
 second symplectic structure given by $\omega=\sum_{1\le i<j}^4\frac{{\rm d} x_i\wedge {\rm d} x_j}{(x_i-x_j)^2}\,$.
Consequently, the system (\ref{CoupledRic}) can be considered, in fact, as a {\it bi--Lie--Hamilton system}.

\subsection{Second-order Kummer--Schwarz equations in Hamiltonian form}
It was proved  in  \cite{CLS12Ham}  that the   second-order Kummer--Schwarz equations \cite{CGL11,BB95}
admit a $t$-dependent Hamiltonian which can be used to work out their Hamilton's equations, namely
\begin{equation}\label{Hamil}
\left\{
\begin{array}{rl}
\frac{{\rm d} x}{{\rm d} t}&=\frac{px^3}{2},\\
\frac{{\rm d} p}{{\rm d} t}&=-\frac{3p^2x^2}{4}- \frac{b_0}{4}+\frac{4b_1(t)}{x^2},
\end{array}\right.
\end{equation}
where $b_1(t)$ is a non-constant $t$-dependent function, $(x,p)\in {\rm T}^*\mathbb{R}_0$ with  $\mathbb{R}_0\equiv \mathbb{R}-\{0\}$, and $b_0$ is a real constant.
 This is  a  Lie system  associated to the $t$-dependent vector field
$X=X_3+b_1(t)X_1$~\cite{CLS12Ham}, where the vector fields
\[
  \qquad X_1=\frac{4}{x^2}\frac{\partial}{\partial p},\qquad
X_2=x\frac{\partial}{\partial x}-p\frac{\partial}{\partial p},\qquad
X_3=\frac{px^3}{2}\frac{\partial}{\partial
x}-\frac 14\left({3p^2x^2}+ b_0\right) \frac{\partial}{\partial p}
\]
span a Vessiot--Guldberg Lie algebra $V$   isomorphic to $\mathfrak{sl}(2,\mathbb{R})$   fulfilling   (\ref{sl2R}).
 Moreover, $X$ is a Lie--Hamilton system, as $V$ consists of
Hamiltonian vector fields with respect to the Poisson bivector
$\Lambda={\partial}/{\partial x}\wedge\partial/ \partial p$ on ${\rm T}^* \mathbb{R}_0$. Indeed,
$X_\alpha=-\widehat\Lambda(dh_\alpha)$, with $\alpha=1,2,3$ and
\begin{equation}\label{FunKS}
h_1=\frac 4 x,\qquad h_2= xp,\qquad h_3=\frac 14\left(p^2x^3+ b_0x \right)
\end{equation}
are a basis of    a Lie algebra isomorphic to $\mathfrak{sl}(2,\mathbb{R})$ satisfying the commutation relations (\ref{sl2Rh}).
Therefore, (\ref{Hamil}) is a Lie--Hamilton system possessing a Lie--Hamiltonian structure
$({\rm T}^*\mathbb{R}_0,\Lambda,h)$, where $h_t=h_3+b_1(t)h_1$. 

To obtain a superposition rule for $X$ we need to determine an integer $m$ so
that the diagonal prolongations  of $X_\alpha$ to ${\rm T}^*\mathbb{R}^m_0$ $(\alpha=1,2,3)$
become linearly independent at a generic point (see \cite{CGM07,CGL12}). This happens for $m=2$. We consider a coordinate system in ${\rm T}\mathbb{R}^3_0$, namely $\{ x_{(1)}, p_{(1)},
x_{(2)}, p_{(2)},x_{(3)}, p_{(3)} \}$. A superposition rule for $X$ can be obtained
by determining two common first integrals for the diagonal prolongations
$\widetilde X_\alpha$ to ${\rm T}^*\mathbb{R}^3_0$
 satisfying
\begin{equation}\label{F1F2}
\frac{\partial(F_1,F_2)}{\partial(x_{(1)},p_{(1)})}\neq 0.
\end{equation}
Instead of searching $F_1,F_2$ in the standard way, i.e.~by solving the system of PDEs given by $\widetilde X_\alpha f =0$,
we make use of Theorem \ref{MT}. This provides such first integrals through the Casimir element $\casS$  (\ref{casSL2}) of the symmetric algebra of
$\ham\simeq {\mathfrak{sl}(2,\mathbb{R})}$. Indeed, the coproduct (\ref{Con}) enables
us to define the elements
\begin{eqnarray*}
  \Delta (C)=\Delta(\ve_1)\Delta(\ve_3)-\Delta(\ve_2)^2=(\ve_1 \otimes  1 + 1 \otimes  \ve_1)(\ve_3 \otimes  1 + 1 \otimes  \ve_3)\!-\!(\ve_2 \otimes
 1 + 1 \otimes  \ve_2)^2,
\\
  \Delta^{(3)} (C)=\Delta^{(3)}(\ve_1)\Delta^{(3)}(\ve_3)-\Delta^{(3)}(\ve_2)^2=(\ve_1 \otimes  1\otimes  1+1 \otimes  \ve_1 \otimes  1 + 1 \otimes  1\otimes
  \ve_1)\times\\  (\ve_3 \otimes  1 \otimes  1 + 1 \otimes  \ve_3 \otimes  1 + 1 \otimes 1 \otimes
  \ve_3)-(\ve_2 \otimes  1 \otimes
 1 + 1 \otimes  \ve_2 \otimes  1 + 1 \otimes 1 \otimes  \ve_2)^2,
\end{eqnarray*}
for $S_{\mathfrak{sl}(2,\mathbb{R})}^{(2)}$ and $S_{\mathfrak{sl}(2,\mathbb{R})}^{(3)}$, respectively. By applying $D$, $D^{(2)}$ and $D^{(3)}$ coming from
   the isomorphism $\phi(\ve_\alpha)=\,h_\alpha$ for the Hamiltonian functions (\ref{FunKS}),     we obtain,  via Theorem \ref{MT},
 the following constants of motion of the type (\ref{invA}):
 \[
\begin{array}{rl}
   &F=D(\casS) = h_1(x_1,p_1) h_3(x_1,p_1)-h_2^2(x_1,p_1)=  b_0,\nonumber \\
  &\FF^{(2)} = D^{(2) } (\Delta(\casS) )
=  \left(h_1(x_1,p_1)+h_1(x_2,p_2)\right)\left(h_3(x_1,p_1)+h_3(x_2,p_2)\right) \nonumber\\   &-\left(h_2(x_1,p_1)+h_2(x_2,p_2)\right)^2
=\frac{b_0(x_1 +x_2)^2+(p_1x_1^2-p_2x_2^2)^2}{x_1x_2}=\frac{b_0(x_1^2+x_2^2)+(p_1x_1^2-p_2x_2^2)^2}{x_1x_2}+2b_0, \nonumber\\
  &\FF^{(3)}  =D^{(3) } (\Delta(\casS) )
= \sum_{i=1}^3 h_1(x_i,p_i) \sum_{j=1}^3h_3(x_j,p_j)-\left( \sum_{i=1}^3 h_2(x_i,p_i)   \right)^2\nonumber\\
  &= \sum_{1\le i<j}^3 \frac{b_0(x_i +x_j)^2+(p_ix_i^2-p_jx_j^2)^2}{x_ix_j} - 3 b_0,
 \nonumber
  \label{intKS}
\end{array}
\]
    where, for the sake of simplicity, hereafter we   denote $(x_i,p_i)$ the coordinates $(x_{(i)},p_{(i)})$. Thus $F$ simply  gives rise to the constant $b_0$, while $\FF^{(2)}$ and $\FF^{(3)}$ are, by construction,   two functionally independent constants of motion for $\widetilde X$ fulfilling (\ref{F1F2}) which, in  turn, allows us to   to derive a superposition rule for $X$. Furthermore, the function $\FF^{(2)}\equiv \FF^{(2)}_{12}$ provides two other constants of the type (\ref{invB}) given by  $\FF_{13}^{(2)}=S_{13}(\FF^{(2)})$ and  $\FF_{23}^{(2)}=S_{23}(\FF^{(2)})$ that verify $  \FF^{(3)}=  \FF^{(2)}+\FF_{13}^{(2)}+\FF_{23}^{(2)}-3b_0 .$
Since it is simpler to work with $\FF_{23}^{(2)}$ than with $  \FF^{(3)}$, we choose the pair $ \FF^{(2)}$, $\FF_{23}^{(2)}$ as the two  functionally independent  first integrals   to obtain a superposition rule. We set
\begin{equation}
\label{xa}
\FF^{(2)}=k_1+2b_0,\qquad \FF_{23}^{(2)}=k_2+2b_0,
\end{equation}
and   compute $x_1,p_1$ in terms of the other variables and $k_1,k_2$.
From (\ref{xa}), we have
\begin{equation}
\label{p1}
p_1=p_1(x_1,x_2,p_2,x_3,p_3,k_1)=\frac{p_2x_2^2\pm\sqrt{k_1 x_1 x_2-b_0(x_1^2+x_2^2)}}{x_1^2}.
\end{equation}
Substituying in the second relation within (\ref{xa}), we obtain
 \begin{equation}
\label{x1}\nonumber
x_1=x_1(x_2,p_2,x_3,p_3,k_1,k_2)= \frac{\aA^2\bB_++b_0 \bB_- (x_2^2-x_3^2) \pm 2 \aA\sqrt{\cC}}{\bB_-^2+4b_0 \aA^2} ,
\end{equation}
provided that the functions $\aA,\bB_\pm,\cC$ are defined by
\[
 \begin{array}{rl}
&\aA=p_2x_2^2-p_3 x_3^2,\qquad \bB_\pm=k_1 x_2\pm k_2 x_3,\\
&\cC=\aA^2\left[ k_1k_2x_2x_3-2b_0^2(x_2^2+x_3^2)-
b_0\aA^2\right]+b_0x_2x_3B_-(k_2x_2-k_1x_3)-b_0^3(x_2^2-x_3^2)^2.
\end{array}
\]
By introducing this result  into (\ref{p1}), we obtain $p_1=p_1(x_2,p_2,x_3,p_3,k_1,k_2)$ which, along with $x_1=x_1(x_2,p_2,x_3,p_3,k_1,k_2)$, provides  a superposition rule for $X$.

In particular for (\ref{Hamil})  with  $b_0=0$ it results
 \[
  x_1= \frac{ \aA^2\left(\bB_+   \pm 2  \sqrt{k_1k_2x_2x_3} \,\right)}{\bB_-^2} ,
\quad p_1= \bB_-^3\frac{\left( \bB_- p_2 x_2^2\pm \aA\sqrt{k_1 x_2\left(\bB_+ \pm 2\sqrt{k_1k_2 x_2 x_3} \,\right)} \,\right) }  {\aA^4\left( \bB_+ \pm 2\sqrt{k_1k_2 x_2 x_3}\, \right)^2} ,
\]
where the functions $\aA,\bB_\pm$ remain in the above same  form. As the constants of motion were derived for non-constant $b_1(t)$, when $b_1(t)$ is constant
we have $V^{X}\subset V$. As a consequence, the functions $F$, $F^{(2)}$, $F^{(3)}$  and so on are still constants of motion for the diagonal prolongation $\widetilde X$ and the superposition rules are still valid for any system  (\ref{Hamil}).

\subsection{Smorodinsky--Winternitz systems with a time-dependent frequency}

Let us focus on the  $n$-dimensional  Smorodinsky--Winternitz systems \cite{WSUF65,WSUF67} with unit mass and  
a non-constant time-dependent frequency $\omega(t)$ whose  Hamiltonian is given by
\[
h=\frac 12  \sum_{i=1}^n p_i^2+ \frac 12 \omega^2(t) \sum_{i=1}^n x_i^2+\frac 12 \sum_{i=1}^n \frac{b_i}{x_i^2},
\]
where the $b_i$'s are $n$ real constants.
  The corresponding Hamilton's equations read
\begin{equation}\label{ClasErm2}
\left\{\begin{array}{rl}
\frac{{\rm d} x_i}{{\rm d} t}&=p_i,\\
\frac{{\rm d} p_i}{{\rm d} t}&=-\omega^2(t)x_i+\frac{b_i}{x_i^3},\\
\end{array}
\right.\qquad i=1,\ldots,n.
\end{equation}
 These systems have been recently attracting quite much attention in classical and
quantum mechanics for their special properties \cite{CLS12Ham,GPS06,SIGMAvulpi,YNHJ11}. Observe that Ermakov systems (\ref{ClasErm}) arise  
as the particular case   of  (\ref{ClasErm2}) for  $n=2$ and $b_2=0$. For $n=1$
the above system maps into the Milne--Pinney equations, which are of interest
in the study of several cosmological models \cite{Dissertationes,Pi50,CL08b}, through the diffeomorphism $(x,p)\in {\rm T}^*\mathbb{R}_0\rightarrow (x,v=p)\in {\rm T}\mathbb{R}_0$.

Let us show that the
system (\ref{ClasErm2}) can be endowed with a Lie--Hamiltonian structure.
This system
 describes the integral curves of the $t$-dependent vector
field on ${\rm T}^*\mathbb{R}^{n}_0$ given by $X=X_3+\omega^2(t)X_1$, where the vector fields
\begin{eqnarray}
  X_1=-\sum_{i=1}^nx_i\frac{\partial}{\partial p_i},\quad
X_2=\frac{1}{2}\sum_{i=1}^n\left(p_i\frac{\partial}{\partial p_i}-x_i\frac{\partial}{\partial x_i}\right), \quad X_3=\sum_{i=1}^n
\left(p_i\frac{\partial}{\partial x_i}+
\frac{b_i}{x_i^3}\frac{\partial}{\partial p_i}\right),
\label{VGSec}\nonumber
\end{eqnarray}
fulfil the commutation rules (\ref{sl2R}). Hence, (\ref{ClasErm2}) is a Lie system. The space ${\rm T}^*\mathbb{R}^n_0$ admits a natural Poisson bivector $\Lambda=\sum_{i=1}^n\partial/\partial x_i\wedge \partial/\partial p_i$ related to the
restriction to this space of the canonical symplectic structure
on ${\rm T}^*\mathbb{R}^n$.  Moreover, the preceding vector
fields are Hamiltonian
vector fields with Hamiltonian functions
\begin{equation}
\label{SWh}
h_1=\frac 12\sum_{i=1}^nx_i^2,\qquad h_2=-\frac 12\sum_{i=1}^nx_ip_i,\qquad h_3=\frac 12\sum_{i=1}^n\left(p_i^2+\frac{b_i}{x_i^2}\right)
\end{equation}
which satisfy the commutation relations (\ref{sl2Rh}), so that $\ham\simeq \mathfrak{sl}(2,\mathbb{R})$.
Consequently, every curve $h_t$ that takes values in the Lie algebra spanned by
$h_1,h_2$ and
$h_3$ gives rise to a Lie--Hamiltonian structure $({\rm
T}^*\mathbb{R}_0^n,\Lambda,h)$. Then,
the system (\ref{ClasErm2}), described by the $t$-dependent vector field
$X=X_3+\omega^2(t)X_1=-\widehat\Lambda(dh_3+\omega^2(t)dh_1),
$ is a Lie--Hamilton system with a
Lie--Hamiltonian structure $({\rm T}^*\mathbb{R}_0^n,\Lambda,
h_t=h_3+\omega^2(t)h_1)$.

Subsequently, we derive an explicit superposition rule for   the  simplest case of the system (\ref{ClasErm2}) corresponding to $n=1$,
and  proceed as in the previous subsection. The prolongations of $X_\alpha$ $(\alpha=1,2,3)$ again become linearly independent for $m=2$ and we need to obtain two first integrals for the diagonal prolongations $\widetilde X_\alpha$ of ${\rm T}^*\mathbb{R}^3_0$ fulfilling (\ref{F1F2}) for the coordinate system $\{ x_{(1)}, p_{(1)},
x_{(2)}, p_{(2)},x_{(3)}, p_{(3)} \}$ of ${\rm T}^*\mathbb{R}^3_0$.
Similarly to the previous example, we have an injection $D:\mathfrak{sl}(2,\mathbb{R})\rightarrow C^\infty({\rm T}^*\mathbb{R}_0)$ which leads to the morphisms   $D^{(2)}$ and $D^{(3)}$. Then, by taking into account the Casimir function (\ref{casSL2}) and the Hamiltonians (\ref{SWh}), we apply Theorem \ref{MT} obtaining the following first integrals:
 \[
\begin{array}{rl}
 \FF^{(2)} = D^{(2) } (\Delta(\casS) )
&= \frac 14
(x_1p_2-x_2p_1)^2+\frac{b(x_1^2+x_2^2)^2}{4x_1^2x_2^2}, \nonumber\\
 \FF^{(3)}  =D^{(3) } (\Delta(\casS) )
&=  \frac 14 \sum_{1\le i<j}^3 \left(
(x_ip_j-x_jp_i)^2+\frac{b(x_i^2+x_j^2)^2}{x_i^2x_j^2} \right)- \frac 34 b,
 \nonumber\\
    \FF_{13}^{(2)}=S_{13}(\FF^{(2)}),\quad  &\quad \FF_{23}^{(2)}=S_{23}(\FF^{(2)}),\qquad   \FF^{(3)}=  \FF^{(2)}+\FF_{13}^{(2)}+\FF_{23}^{(2)}-3b/4 ,
   \label{intSW}
\end{array}
\]
    where  $(x_i,p_i)$ denote the coordinates $(x_{(i)},p_{(i)})$; notice that   $F=D(\casS) =    { b}/4$.
 We choose $\FF^{(2)}$ and $ \FF_{23}^{(2)}$ as the two functionally independent constants of motion and we shall use $ \FF_{13}^{(2)}$ in order to simplify the results. Recall that
 these functions are exactly the first integrals obtained in other works, e.g.~\cite{SIGMA}, for describing superposition rules of dissipative Milne--Pinney
equations (up to the diffeomorphism
$\varphi:(x,p)\in {\rm T}^*\mathbb{R}_0\mapsto (x,v)=(x,p)\in {\rm
T}\mathbb{R}_0$ to system (\ref{ClasErm2}) with $n=1$), and lead straightforwardly to
deriving a superposition rule for these equations \cite{CL08b}.

 Indeed, we set
 \begin{equation}
\label{ff1}
 \FF^{(2)}=\frac{k_1}4+\frac{b}2 ,\qquad   \FF_{23}^{(2)}=\frac{k_2}4+\frac{b}2 ,\qquad   \FF_{13}^{(2)}=\frac{k_3}4+\frac{b}2 ,\qquad
\end{equation}
 and from the first equation we
    obtain
 $p_1$ in terms of the remaining variables and $k_1$:
\begin{equation}
\label{pp1}
p_1=p_1(x_1,x_2,p_2,x_3,p_3,k_1)=\frac{p_2x_1^2x_2\pm\sqrt{k_1 x_1^2 x_2^2-b(x_1^4+x_2^4)}}{x_1 x_2^2}.
\end{equation}
By  introducing this value in the second expression  of (\ref{ff1}),
 one can determine the expression of $x_1$ as a function of $x_2,p_2,x_3,p_3$ and the constants $k_1,k_2$. Such a result is rather simplified when the third constant of  (\ref{ff1}) enters, 
 yielding
 \begin{equation}
\begin{array}{rl}
x_1&=x_1(x_2,p_2,x_3,p_3,k_1,k_2)=x_1(x_2,x_3,k_1,k_2,k_3)\\
&=\left[ {\MM_1 x_2^2+\MM_2 x_3^2\pm \sqrt{\MM \left[k_3 x_2^2 x_3^2 -b (x_2^4+x_3^4) \right]}}    \right]^{1/2} ,
 \end{array}
 \label{xx1}
\end{equation}
where the constants $\MM_1,\MM_2,\MM$ are defined in terms of $k_1,k_2,k_3$ and $b$  as
\[
  \MM_1=\frac{2b k_1-k_2k_3}{4b^2-k_3^2},\qquad \MM_2=\frac{2b k_2-k_1k_3}{4b^2-k_3^2},\qquad
\MM=\frac{4\left[4b^3 +k_1 k_2 k_3 - b(k_1^2+k_2^2+k_3^2) \right]}{(4b^2-k_3^2)^2}.
\]
 And by introducing (\ref{xx1}) into (\ref{pp1}), we   obtain $p_1=p_1(x_2,p_2,x_3,p_3,k_1,k_2)=p_1(x_2,p_2,x_3,k_1,k_2,k_3),$
  which together with   (\ref{xx1})  provide a superposition rule for (\ref{ClasErm2}) with $n=1$.
  These expressions constitute the known superposition rule for Milne--Pinney equations \cite{CL08b}.
Observe that,
instead of solving systems of PDEs for obtaining  the first integrals   as
in \cite{SIGMA,CL08b}, we have obtained them algebraically in a simpler way. When $b=0$ we recover, as expected, 
the superposition rule for the harmonic oscillator with a $t$-dependent frequency. Similarly to previous examples, the above
superposition rule is also valid when $\omega(t)$ is constant.

\subsection{A classical system with trigonometric nonlinearities}
Let us study a final example appearing in the
study of integrability of classical systems \cite{ADR11,AW05}. Consider the
system
\[
\left\{\begin{array}{rl}
\frac{{\rm d} x}{{\rm d} t}&=\sqrt{1-x^2}\left(B_x(t)\sin\,p-B_y(t)\cos p\right),\\
\frac{{\rm d} p}{{\rm d} t}&=-(B_x(t) \cos p+B_y (t)\sin p)\frac{x}{\sqrt{1-x^2}}-B_z(t),
\end{array}\right.
\]
where $B_x(t),B_y(t), B_z(t)$ are arbitrary $t$-dependent
functions and $(x,p)\in {\rm T}^*{\rm I}$, with ${\rm I}=(-1,1)$. This system describes the integral curves of the $t$-dependent vector
field
\begin{eqnarray*}
  X=\sqrt{1-x^2}(B_x(t)\sin\,p\!-\!B_y(t)\cos p)\frac{\partial}{\partial x}\!-\!\left[\frac{(B_x(t) \cos
p\!+\!B_y (t)\sin p)x}{\sqrt{1-x^2}}\!+\!B_z(t)\right]\frac{\partial}{\partial p},
\end{eqnarray*}
which can be brought into the form $X=B_x(t)X_1+B_y(t)X_2+B_z(t)X_3$, where
\begin{eqnarray*}
  X_1=\sqrt{1-x^2}\sin\, p\frac{\partial}{\partial x}\!-\!\frac{x}{\sqrt{1-x^2}}\cos
p\frac{\partial}{\partial p},\,\,
X_2=-\sqrt{1-x^2}\cos\,p\frac{\partial}{\partial x}\!-\!\frac{x}{\sqrt{1-x^2}}\sin
p\frac{\partial}{\partial p},
\end{eqnarray*}
and $X_3=-\partial/\partial p$ satisfy the commutation relations
\[
 [X_1,X_2]=X_3, \qquad [X_3,X_1]=X_2, \qquad [X_2,X_3]=X_1.
\]
In other words, $X$ describes a Lie system associated with a Vessiot--Guldberg
Lie algebra isomorphic to $\mathfrak{su}(2)$. As in the previous examples, we assume $V^X=V$. Now, the vector fields
 $X_\alpha$ $(\alpha=1,2,3)$ are Hamiltonian ones with Hamiltonian functions given by
\begin{equation}
\label{so3}
h_1=-\sqrt{1-x^2}\cos p,\qquad h_2=-\sqrt{1-x^2}\sin p,\qquad h_3=x,
\end{equation}
thus spanning a real Lie algebra isomorphic to $\mathfrak{su}(2)$. Indeed,
\[
 \{h_1,h_2\}=-h_3 ,\qquad \{h_3,h_1\}=-h_2 ,\qquad \{h_2,h_3\}=-h_1 .
\]
Next we consider a basis $\{\ve_1,\ve_2,\ve_3\}$
for $\mathfrak{su}(2)$ satisfying
\[
[\ve_1,\ve_2]=-\ve_3,\qquad  [\ve_3,\ve_1]=-\ve_2,\qquad [\ve_2,\ve_3]=-\ve_1,
\]
so that  $\mathfrak{su}(2)$ admits the Casimir $\casU= \ve_1\widetilde\otimes \ve_1+\ve_2\widetilde\otimes \ve_2+\ve_3\widetilde\otimes \ve_3\in  U_{\mathfrak{su}(2)}$. Then,   the Casimir element of $S_{\mathfrak{su}(2)}$ reads 
$
\casS=\lambda^{-1}(\casU)=\ve_1^2+ \ve^2_2 + \ve^2_3 .
$

The diagonal prolongations of $X_1,X_2,X_3$ are linearly independent at a generic point for $m=2$ and we have to derive
two first integrals for the diagonal prolongations $\widetilde X_1,\widetilde X_2,\widetilde X_3$ on ${\rm T}^*{\rm I}^3$ satisfying (\ref{F1F2}) 
working with the coordinates $\{ x_{(1)}, p_{(1)}, x_{(2)}, p_{(2)},x_{(3)}, p_{(3)} \}$ of ${\rm T}^*{\rm I}^3$. Then, by taking into account the Casimir function $C$, 
the Hamiltonians (\ref{so3}), the isomorphism $\phi(\ve_\alpha)=\,h_\alpha$ and the injection $D:\mathfrak{sl}(2,\mathbb{R})\rightarrow C^\infty({\rm T}^*{\rm I})$, we apply Theorem \ref{MT} obtaining the following first integrals:
 \[
\begin{array}{rl}
  \FF^{(2)}
&= 2\left(\sqrt{1-x^2_1}\sqrt{1-x^2_2}\cos (p_1-p_2)+x_1x_2+1\right) , \nonumber\\
  \FF^{(3)}
&=  2  \sum_{1\le i<j}^3 \left(
\sqrt{1-x^2_i}\sqrt{1-x^2_j}\cos (p_i-p_j)+x_ix_j \right)+3,
 \nonumber\\
    \FF_{13}^{(2)}&=S_{13}(\FF^{(2)}),   \qquad \FF_{23}^{(2)}=S_{23}(\FF^{(2)}),\qquad   \FF^{(3)}=  \FF^{(2)}+\FF_{13}^{(2)}+\FF_{23}^{(2)}-3  ,
 \end{array}
\]
and  $F=D(\casS) =   1$.
 We again choose $\FF^{(2)}$ and $ \FF_{23}^{(2)}$ as the two functionally independent constants of motion,
  which provide us, after   cumbersome but straightforward computations, with a superposition rule
for these systems.
This leads to a quartic equation, whose solution can be obtained
through known methods. All our results are also valid for the case when $V^X\subset V$.

\section{Conclusions and outlook}
 We have proved several new properties of the constants of motion for Lie--Hamilton systems. New methods for their calculation have been devised, and Poisson coalgebra techniques have been developed for obtaining superposition rules.

Our achievements strongly simplify the search for constants of motion and superposition rules by avoiding the integration of ODEs and PDEs required by standard methods. 
We have provided generalisations of previous results on Lie--Hamilton systems \cite{CLS12Ham,Ru10} and coalgebra integrability of autonomous Hamiltonian systems \cite{BBHMR09,BR}. Finally, we illustrated our approach by analysing several non-autonomous systems of interest.

In the future, we aim to apply our techniques to new relevant systems. Furthermore, we expect to extend our 
formalism to Lie systems admitting 
a Vessiot--Guldberg Lie algebra
of Hamiltonian vector fields with respect to a Dirac structure. This would enable us to use our procedures to study
a broader variety of systems. Moreover, we also expect to analyse the use of Poisson coalgebra techniques 
to devise an algebraic approach to Lie--Hamilton systems with mixed superposition rules.

\section*{Acknowledgements}

The research of J. Cari\~nena, J. de Lucas and C. Sard\'on was
supported by the Polish National Science Centre
under the grant HARMONIA Nr 2012/04/M/ST1/00523
and the research project MTM--2009--11154.
J. de Lucas would like to thank a grant FMI40/10 from the Direcci\'on General de Arag\'on
to perform a research stay in
Zaragoza. C. Sard\'on acknowledges a fellowship provided by the
University of Salamanca and partial financial support by research proyect
FIS2009-07880
(DGICYT).   A.~Ballesteros and F.J.~Herranz acknowledge partial financial suport from the project MTM--2010--18556.

\appendix
\section*{Appendix}
\setcounter{section}{1}

We now detail the proof of Lemma \ref{62}. The elements of $S_\mathfrak{g}$ given by $v^\JJ\equiv v_1^{i_1}\cdot\ldots\cdot v_r^{i_r}$,
where the $\JJ$'s are $r$-multi-indexes, ``$\cdot$" denotes the product of elements of $\mathfrak{g}$ as functions on $\mathfrak{g}^*$ and $\{v_1,\ldots,v_r\}$ is a basis for $\mathfrak{g}$, form a basis of $S_\mathfrak{g}$.
 Then, every $P\in S_\mathfrak{g}$ can be written in a unique way as $P=\sum_{\JJ\in M}\lambda_\JJ v^\JJ$, where $M$ is a finite family of multi-indexes and each $\lambda_\JJ\in\mathbb{R}$. 
 Hence, the $\mathbb{R}$-algebra morphism  $D:(S_\mathfrak{g},\cdot)
\rightarrow
(C^\infty(N),\cdot)$ extending $\phi : \mathfrak{g}\rightarrow \ham$ is determined by the image of the elements of a basis for $\mathfrak{g}$. Indeed,
\begin{equation}\label{dec2}
D(P)=\sum_\JJ\lambda_\JJ D(v^\JJ)=\sum_\JJ\lambda_\JJ\phi(v^{i_1}_1)\cdot\ldots\cdot \phi(v^{i_r}_r).
\end{equation}
Let us prove that $D$ is also an $\mathbb{R}$-algebra morphism. From (\ref{dec2}),
we see that $D$ is linear. Moreover, $D(PQ)=D(P)D(Q)$
for every $P,Q\in S_\mathfrak{g}$. In fact, if we write $Q=\sum_{\KK\in M}\lambda_\KK v^\KK$, we obtain
\begin{eqnarray*}
  D(PQ)\!=\!\!D\!\left(\!\sum_\JJ\!\lambda_\JJ v^\JJ\!\sum_\KK\!\lambda_\KK v^\KK\!\!\right)\!\!=\!\!\sum_\LL
\!\!\sum_{\JJ\!+\!\KK=\!\LL}\!\!\!\!\lambda_\JJ\lambda_\KK D(v^\LL)\!=\!\!\sum_\JJ\!\lambda_\JJ D(v^\JJ)\!\!\sum_\KK\!\lambda_\KK D(v^\KK)\!\!=\!\!D(P)D(Q),\nonumber
\end{eqnarray*}
where $\JJ+\KK=(i_1+j_1,\ldots,i_r+j_r)$ with $\JJ=(i_1,\ldots,i_r)$ and $\KK=(j_1,\ldots, j_r)$.

Let us show that  $D$ is also a Poisson morphism. By linearity, this reduces to proving that
$D\left ( \{v^\JJ,v^\KK\}_{S_\mathfrak{g}}\right)=\{D(v^\JJ),D (v^\KK)\}_\Lambda$ for arbitrary $\JJ$ and $\KK$. Define $|\JJ|=i_1+\ldots+i_r$.
If $|\KK|=0$ or $|\JJ|=0$ this is satisfied, as a Poisson bracket vanishes when any entry is a constant. We now prove by induction 
the remaining cases. For $|\JJ|+|\KK|=2$, we have 
\[
  {D}\left( \{v_\alpha,v_\beta\}_{S_\mathfrak{g}} \right)={\phi}([v_\alpha,v_\beta]_{\mathfrak{g}})=\{ \phi (v_\alpha),\phi (v_\beta)\}_\Lambda=\{{D} (v_\alpha),{D}
(v_\beta)\}_\Lambda ,\quad \forall \alpha,\beta=1,\ldots,r.
\]
If $D$ is a Poisson morphism for $|\JJ|+|\KK|=m>2$,
then for $|\JJ|+|\KK|=m+1$ we can set $v^\JJ=v^{\bar \JJ}v^{i_{\gamma}}_\gamma$ for $i_\gamma\neq 0$ and some $\gamma$
to obtain
\begin{eqnarray*}
 && D\left(\{v^\JJ,v^\KK\}_{S_\mathfrak{g}}\right)=D\left(\{v^{\bar \JJ}v_\gamma^{i_\gamma},v^\KK\}_{S_\mathfrak{g}}\right)=D\left(\{v^{\bar \JJ},v^\KK\}_{S_\mathfrak{g}}
v_\gamma^{i_\gamma}+v^{\bar \JJ}\{v_\gamma^{i_\gamma},v^\KK\}_{S_\mathfrak{g}}\right)\cr
  &&\qquad \qquad\qquad \quad\!
 =
\{D(v^{\bar \JJ}),D(v^\KK)\}_\Lambda D (v_\gamma^{i_\gamma})+
D( v^{\bar \JJ})\{D(v_\gamma^{i_\gamma}),D (v^\KK)\}_\Lambda\cr
  &&\qquad\qquad\qquad\quad\!	 =\{D( v^{\bar \JJ})D (v_\gamma^{i_\gamma}),D (v^\KK)\}_\Lambda=
\{D(v^\JJ),D( v^\KK)\}_\Lambda.\nonumber
\end{eqnarray*}
By induction, $D\left ( \{v^\JJ,v^\KK\}_{S_\mathfrak{g}}\right)=\{D(v^\JJ),D (v^\KK)\}_\Lambda$ for any $I$ and $J$.


\begin{thebibliography}{10}
\bibitem{Ol93}
Olver PJ 1993
\newblock {\it Applications of {L}ie Groups to Differential Equations} (Graduate Texts in Mathematics vol 107)
\newblock (New York: Springer-Verlag)

\bibitem{WTC83}
Weiss J, Tabor M and Carnevale G. 1983
\newblock {\it J. Math. Phys.} {\bf 24}(3) 522--26.

\bibitem{Be84}
Berry M V  1984
\newblock {\it Proc. Roy. Soc. London Ser.} {\rm A} {\bf 392} 45--57


\bibitem{LS}
Lie S and Scheffers G 1983
\newblock {\em Vorlesungen {\"u}ber continuierliche {G}ruppen mit Geometrischen
  und Anderen {A}nwendungen}
\newblock (Leipzig: Teubner)

\bibitem{CGM00}
Cari{\~n}ena J F, Grabowski J and Marmo G 2000
\newblock {\em {L}ie--{S}cheffers Systems: a Geometric Approach}
\newblock (Naples: Bibliopolis)

\bibitem{PW}
Winternitz P 1983
\newblock {\it Nonlinear phenomena ({O}axtepec, 1982)} vol 189 of {\em
  Lecture Notes in Phys.} (Berlin: Springer) p 263--331

\bibitem{CGM07}
Cari{\~n}ena J F, Grabowski J and Marmo G  2007
\newblock {\it Rep. Math. Phys.} {\bf 60}(2) 237--258.

\bibitem{Ib99}
Ibragimov N H 1999
\newblock {\em Elementary Lie group analysis and ordinary differential
  equations} (Wiley Series in Mathematical Methods in Practice vol 4)
\newblock (Chichester:  John Wiley \& Sons Ltd.)

\bibitem{Dissertationes}
Cari{\~n}ena J F and de~Lucas J 2011
\newblock {\it Dissertationes Math. (Rozprawy Mat.)} {\bf 479} 1--162


\bibitem{SIGMA}
Cari\~nena J F, de Lucas J and Ra{\~n}ada M F 2008
\newblock {\it SIGMA} {\bf 4} 031

\bibitem{WintSecond}
Rogers C, Schief W K and Winternitz P 1997
\newblock {\it J. Math. Anal. Appl.} {\bf 216}(1) 246--64

\bibitem{FLV10}
Flores-Espinoza R, de Lucas J and Vorobiev Y M 2010
\newblock {\it J. Phys. {\rm A}: Math. Teor.} {\bf 43}(20) 205208

\bibitem{CGL12}
Cari\~nena J F, Grabowski J and de Lucas J 2012
\newblock {\it J. Phys. {\rm A}: Math. Theor.} {\bf 45}(2) 185202



\bibitem{Clem06}
Clemente-Gallardo J 2006
\newblock {\it Groups, geometry and physics}, (Monogr. Real Acad. Ci. Exact.
  F\'\i s.-Qu\'\i m. Nat. Zaragoza vol 29) (Zaragoza: Acad. Cienc. Exact. F\'\i
  s. Qu\'\i m. Nat. Zaragoza) p 65--78

\bibitem{cal}
Avram F, Cari\~nena J F and de Lucas J 2010
\newblock {\it Modern Trends of Controlled
  Stochastic Processes: Theory and Applications} ed A B Piunovskiy (Manchester: Luniver Press) p 144-60

\bibitem{CLS12Ham}
Cari\~nena J F, de Lucas J and Sard\'on C 2013
\newblock Lie--{H}amilton systems: theory and applications.
\newblock {\it Int. J. Geom. Methods Mod. Physics} ({\it Preprint} arxiv:1211.6991)

\bibitem{CLS12}
Cari\~nena J F, de Lucas J and Sard\'on C 2012
\newblock {\it Int. J. Geom. Method. Mod. Physics.} {\bf 9}(2) 1260007

\bibitem{ADR11}
Angelo R M, Duzzioni E I, and Ribeiro A D 2012
\newblock {\it J. Phys. {\rm A}: Math. Theor.} {\bf 45}(5) 055101

\bibitem{BBHMR09}
Ballesteros A, Blasco A, Herranz F J, Musso F and Ragnisco O 2009
\newblock {\it J. Phys. Conf. Ser.} {\bf 175} 012004

\bibitem{Ru10}
Flores-Espinoza R 2011
\newblock {\it Int. J. Geom. Methods Mod. Phys.} {\bf 8}(6) 1169--77


\bibitem{CP}
Chari V  and Pressley A 1994
\newblock
  {\it A Guide to Quantum Groups} (Cambridge: Cambridge University Press)




\newblock {\it Czechoslovak J. Phys.} {\bf 46}(12) 1153--63













\bibitem{WSUF65}
Fris J, Mandrosov V, Smorodinsky Y A, Uhlir M and Winternitz P 1965
{\it  Phys. Lett.} {\bf 16}(3) 354--56

\bibitem{WSUF67}
Winternitz P, Smorodinsky A, Uhlir M and Fris J 1967
\newblock {\it Soviet J. Nuclear Phys.} {\bf 4} 444--50

\bibitem{Pi50}
Pinney E 1950
\newblock {\it Proc. Amer. Math. Soc.} {\bf 1} 681

\bibitem{Er08}
Ermakov V P 2008
\newblock {\it Appl. Anal. Discrete Math.} {\bf 2}(2) 123--45

\bibitem{BR} Ballesteros A and Ragnisco O 1998
 {\it J. Phys. {\rm A}: Math. Gen.}, {\bf 31}(16) 3791--813



\bibitem{Foundations}
Abraham R and Marsden J E 1987
\newblock {\it Foundations of mechanics}.
\newblock (Redwood City: Addison-Wesley Publishing Company Inc.)

\bibitem{CR89}
Cari\~nena J F and Ra{\~n}ada M F 1989
\newblock {\it J. Math. Phys.} {\bf 30}(10) 2258--66




\bibitem{CL99}
Cari\~nena J F and L{\'o}pez C 1999
\newblock {\it Rep. Math. Phys.} {\bf 43}(1-2) 43--51.

\bibitem{Be81}
Berdjis F 1981
\newblock {\it J. Math. Phys.}  {\bf 22} 1851--56


\bibitem{AMA75}
Abellanas L and Mart{\'{\i}}nez Alonso L 1975
\newblock {\it J. Math. Phys.} {\bf 16} 1580--84

\bibitem{Var98}
Varadarajan V S 1984
\newblock {\em Lie groups, {L}ie algebras, and their representations} (Graduate Texts in Mathematics vol 102)
\newblock (New York: Springer-Verlag)

\bibitem{Va94}
Vaisman I 1994
\newblock {\em Lectures on the geometry of {P}oisson manifolds}
  (Progress in Mathematics vol 18)
\newblock (Basel: Birkh\"auser Verlag)

\bibitem{CRC96}
Anderson R L, Baikov V A {\it et al.} 1996
\newblock {\it C{RC} handbook of {L}ie group analysis of differential
  equations. vol 3. New trends in theoretical developments and computational methods}
\newblock (Boca Raton, FL: CRC Press)




\bibitem{AHW81}
Anderson R L and Harnad J 1981
\newblock {\it Lett. Math. Phys.}  {\bf 5}(2) 143--48

\bibitem{CL11Sec}
Cari\~nena J F and de Lucas J 2011
\newblock {\it J. Geom. Mech.} {\bf 3}(1) 1--22





\bibitem{Schrod}
Ballesteros A, Herranz  F J and Parashar P 2000
{\it J. Phys. {\rm A}: Math. Gen.} {\bf 33}(17) 3445--65

\bibitem{BBF09}
Ballesteros A, Blasco A and Herranz F J 2009
\newblock {\it J. Phys. {\rm A}: Math. Theor.} {\bf 42}(26) 265205


\bibitem{Ma95}
Maamache M 1995
\newblock {\it Phys. Rev. {\rm A}} {\bf 52}(2) 936--40

\bibitem{LT05}
Luan P G  and Tang C S 2005
\newblock {\it Phys. Rev. {\rm A}} {\bf 71}(1) 014101

\bibitem{EV05}
Espinoza R F and Vorobiev Y M 2005
\newblock {\it Russ. J. Math. Phys.} {\bf 12}(3) 326--49

\bibitem{CarRamGra}
Cari\~nena J F, Grabowski J and Ramos A 2001
\newblock {\it Acta Appl. Math.} {\bf 66}(1) 67--87

\bibitem{Po86}
Postnikov M 1986
\newblock {\em Lie Groups and {L}ie Algebras. Lectures in Geometry. Semester V}
\newblock (Moscow: Mir)

\bibitem{GO64}
Gruber B and O'Raifeartaigh L 1964
\newblock {\it J. Math. Phys.} {\bf 5}(12) 1796--1804


\bibitem{PP68}
Perelomov A M and Popov V S 1968
\newblock {\it Math. USSR--Izvestija} {\bf 2} (2) 1313--35

\bibitem{Vulpi}
Herranz F J and Santander M 1997
{\it J. Phys. {\rm A}: Math. Gen.} {\bf 30}(15) 5411--26


\bibitem{invariantsWint}
Patera J, Sharp R T and Winternitz P 1976
\newblock {\it J. Math. Phys} {\bf 17} (6) 986--94

\bibitem{campoamor}
Campoamor-Stursberg R 2002
{\it J. Phys. {\rm A}: Math. Gen.} {\bf 35} (30) 6293--306

\bibitem{campoamor1}
Ancochea J M, Campoamor-Stursberg R and Garcia Vergnolle L 2006
{\it J. Phys. {\rm A}: Math. Gen.} {\bf 39}(14) 1339--55

\bibitem{campoamor2}
Campoamor-Stursberg R 2010
{\it J. Phys. {\rm A}: Math. Theor.} {\bf 43} (14) 145202


\bibitem{BCR96}
Ballesteros A, Corsetti M and Ragnisco O 1996








\bibitem{LA08}
Leach P G L and Andriopoulos K 2008
\newblock {\it Appl. Anal. Discrete Math.} {\bf 2}(2) 146--57

\bibitem{Be97}
Berkovich L M 1997
\newblock    {\it Symmetry in nonlinear mathematical physics}, {V}ol. 1, 2 (Kiev: Natl. Acad. Sci. Ukraine, Inst. Math.) p 155--63


\bibitem{AW05}
Angelo R M and Wreszinski W F 2005
\newblock {\it Phys. Rev. {\rm A}} {\bf 72} 034105

\bibitem{CGL11}
Cari\~nena J F, Grabowski J and de Lucas J 2012
\newblock {\it J. Phys. {\rm A}: Math. Theor.}  {\bf 45}(18) 185202

\bibitem{BB95}
Berkovich L M and Berkovich F L 1995
\newblock {\it Univ. Beograd. Publ. Elektrotehn. Fak. Ser. Mat.} {\bf 6} 11--24

\bibitem{GPS06}
Grosche C, Pogosyan G S and Sissakian A N 1995
\newblock {\it Fortschr. Phys.} {\bf 43}(6) 453--521



\bibitem{SIGMAvulpi}
Herranz F J and Ballesteros A 2006
\newblock {\it SIGMA} {\bf 2} 010


\bibitem{YNHJ11}
You-Ning L and Hua-Jun H 2011
\newblock {\it Chinese Phys. {\rm B}} {\bf 20} 010302




\bibitem{CL08b}
Cari\~nena J F and de Lucas J 2008
\newblock {\it Phys. Lett. {\rm A}} {\bf 372}(33) 5385--89

\end{thebibliography}
\end{document}